\documentclass[pra,aps,superscriptaddress,twocolumn,letter,nopacs,nofootinbib,longbibliography,notitlepage]{revtex4-2}
\usepackage{graphicx,color,amsmath,amsfonts,enumerate,amsthm,amssymb,mathtools,enumitem,thmtools,hyperref,mathdots,enumitem,centernot,bm,soul,bbm,stmaryrd}
\usepackage[capitalise, noabbrev]{cleveref}
\usepackage{multirow}
\usepackage{amsmath}
\usepackage{times}
\usepackage{tikz,pgfplots}
\usetikzlibrary{quantikz}
\hypersetup{colorlinks=true,linkcolor=blue,citecolor=blue,urlcolor=blue}
\usepackage{bbm}
\usepackage[caption=false]{subfig}
\usepackage{algorithm,algpseudocode}
\usepackage{algorithmicx}
\usepackage[skins]{tcolorbox}
\newcommand{\commentblock}[1]{}
\usepackage[normalem]{ulem}

\theoremstyle{plain}
\newtheorem{thm}{Theorem}
\newtheorem{lem}[thm]{Lemma}

\newtheorem{cor}[thm]{Corollary}

\theoremstyle{definition}

\definecolor{ppblue}{RGB}{46,117,182}
\definecolor{ppred}{RGB}{197, 90, 17}

\definecolor{hpCol}{RGB}{150, 0, 150}

\newcommand{\bbc}{\mathbb{C}}

\newcommand{\abs}[1]{\left| {#1} \right|} 
\newcommand{\tr}[1]{\mathrm{tr}\left( #1 \right)}

\newcommand{\prover}[2]{\underset{#1}{\rm Pr}\left(#2\right)}
\newcommand{\expval}[2]{\underset{#1}{\mathbb{E}}\left[{#2}\right]}

\newcommand{\set}[1]{\left\{#1\right\}}
\newcommand{\bitstring}[1]{\left\{0,1\right\}^{#1}}
\newcommand{\order}[1]{\mathcal O\left(#1\right)}

\newcommand{\pnorm}[2]{||#1||_{#2}}

\newcommand{\ketbra}[2]{\ensuremath{\left|#1\right\rangle\!\!\left\langle#2\right|}}

\newcommand{\matrixel}[3]{\ensuremath{\left\langle #1 \vphantom{#2#3} \right| #2 \left| #3 \vphantom{#1#2} \right\rangle}}

\newcommand{\pauli}[1]{\mathcal{P}_{#1}}
\newcommand{\clif}[1]{\mathcal{C}_{#1}}

\newcommand{\logicnor}[1]{\textsc{nor}(#1)}
\newcommand{\logicor}[1]{\textsc{or}(#1)}
\newcommand{\ensemble}{\mathcal{U}}
\newcommand{\sensemble}[1]{\mathcal{U}_{#1}}
\newcommand{\defensemble}{(\mathbb{U},\mu)}
\newcommand{\defsensemble}[1]{(\mathbb{U}_{#1},\mu_{#1})}
\newcommand{\eset}{\mathbb{U}}
\newcommand{\eprob}{\mu}
\newcommand{\seset}[1]{\mathbb{U}_{#1}}
\newcommand{\seprob}[1]{\mu_{#1}}
\newcommand{\measmap}[1]{\mathcal{M}(#1)}
\newcommand{\smeasmap}[2]{\mathcal{M}_{#2}(#1)}
\newcommand{\invmeasmap}[1]{\mathcal{M}^{-1}(#1)}
\newcommand{\invsmeasmap}[2]{\mathcal{M}^{-1}_{#2}(#1)}
\newcommand{\shadownorm}[1]{\left\lVert #1 \right\rVert_{\mathrm{s}}}
\newcommand{\sshadownorm}[2]{\left\lVert #1 \right\rVert_{\mathrm{s}(#2)}}
\newcommand{\ssshadownorm}[3]{\left\lVert #1 \right\rVert_{\mathrm{s}(#2),#3}}
\newcommand{\wt}[1]{\abs{#1}}
\newcommand{\tprob}[1]{t_d(#1)}
\newcommand{\tprobl}[1]{t_{#1,d}}
\newcommand{\stprobl}[2]{t_{#1,#2}}
\newcommand{\ttprobl}[1]{\tau_{(#1),d}}
\newcommand{\stprob}[2]{t_{#2}(#1)}

\begin{document}

% -------------------------------------------------------------
% TITLE & ABSTRACT
% -------------------------------------------------------------

\title{Shallow shadows: Expectation estimation using low-depth random Clifford circuits}
%\title{Shadow estimation with log-depth random quantum circuits}
% \date{\today}

\author{Christian Bertoni}
\affiliation{Dahlem Center for Complex Quantum Systems, Freie Universit\"{a}t Berlin, Germany}
\author{Jonas Haferkamp}
\affiliation{Dahlem Center for Complex Quantum Systems, Freie Universit\"{a}t Berlin, Germany}
\author{Marcel Hinsche}
\affiliation{Dahlem Center for Complex Quantum Systems, Freie Universit\"{a}t Berlin, Germany}
\author{Marios Ioannou}
\affiliation{Dahlem Center for Complex Quantum Systems, Freie Universit\"{a}t Berlin, Germany}
\author{Jens Eisert}
\affiliation{Dahlem Center for Complex Quantum Systems, Freie Universit\"{a}t Berlin, Germany}
\address{Helmholtz-Zentrum Berlin f{\"u}r Materialien und Energie, 14109 Berlin, Germany}
\address{Fraunhofer Heinrich Hertz Institute, 10587 Berlin, Germany}
\author{Hakop Pashayan}
\affiliation{Dahlem Center for Complex Quantum Systems, Freie Universit\"{a}t Berlin, Germany}
\date{\today}
\begin{abstract}
We provide practical and powerful schemes for learning many properties of an unknown $n$-qubit quantum state using a sparing number of copies of the state. Specifically, we present a depth-modulated randomized measurement scheme that interpolates between two known classical shadows schemes based on \emph{random Pauli measurements} and \emph{random Clifford measurements}. These can be seen within our scheme as the special cases of zero and infinite depth respectively. We focus on the regime where depth scales logarithmically in $n$ and provide evidence that this retains the desirable properties of both extremal schemes whilst, in contrast to the random Clifford scheme, also being experimentally feasible. We present methods for two key tasks; estimating expectation values of certain observables from generated classical shadows and, computing upper bounds on the depth-modulated shadow norm, thus providing rigorous guarantees on the accuracy of the output estimates. We consider observables that can be written as a linear combination of $\mathrm{poly}(n)$ Paulis and observables that can be written as a low bond dimension matrix product operator. For the former class of observables both tasks are solved efficiently in $n$. For the latter class, we do not guarantee efficiency but present a 
tensor network method that works in practice; by variationally computing a heralded approximate inverse of a tensor network.
\end{abstract}

\maketitle

\section{Introduction}
\label{sec:intro}
It is a well known fact that learning a classical description of a quantum state from access to a limited number of samples (copies of the state) is a demanding task for quantum systems involving a large number of degrees of freedom. Quantum state tomography produces an extremely accurate classical description of the unknown quantum states, but consumes an enormous number of samples \cite{ODonnel2016efficient, Haah2016sample, Tomography}. 
Recently, focus has shifted towards sample efficient approaches that satisfy weaker standards of state learning but are nevertheless very useful. Huang, Kueng and Preskill’s (HKP) seminal work~\cite{Huang2020}, also known as the classical shadows protocol,
%, building on ideas from ~\cite{Aaronson2018shadow} 
is a state learning protocol based on randomized measurements~\cite{ohliger_efficient_2013,RenyiEntropies,PhysRevLett.120.050406,RandomMeasurementsToolbox}. This protocol allows one to estimate the expectation values of many observables from the outcomes of suitable randomized measurements of the unknown state. Importantly, in contrast to the related work of  Ref.~\cite{Aaronson2018shadow}, these randomized measurements are obtained \emph{obliviously} i.e. without prior knowledge of the specific choice of observables. 
 The scheme comes with rigorous performance guarantees, permitting one to determine the trade-off between the number of samples and the accuracy of the estimated expectation value for any given observable.
Classical shadows have found many applications including estimating: expected molecular energies~\cite{Hadfield2022, hadfield2021adaptive}, the purity and ground state proximity of 
\emph{spin chains}~\cite{notarnicola2021} and gate-set properties~\cite{helsen2021estimating} for noise characterization. It has been used to detect entanglement~\cite{Elben2020,Neven2021,Rath2021} and chaos~\cite{Garcia2021} in many body systems; classify quantum data~\cite{li2021} and develop improved variational search algorithms~\cite{Boyd2022,sack2022}. For example,
by relying on the scheme's ability to obliviously estimate the expectation values of many observables, Ref.~\cite{Boyd2022} has proposed a promising new method for training variational circuits to find Hamiltonian eigenstates and Ref.~\cite{Huggins2022} improved existing computational methods and used them to estimate ground state atomization energies of molecules.

Classical shadows have seen some experimental implementation~\cite{struchalin2021experimental,elben2020mixed,zhang2021experimental}, inspired substantial further study~\cite{Lukens2021,Chen2022,Acharya2021InformationallyCP, Hadfield2022} and many related proposals and generalizations including; noise-robustness analysis and noise-robust variants~\cite{Chen2021,Koh2022classicalshadows,flammia2021averaged}, variants for fermionic systems~\cite{Zhao2021} and performance improvement in the setting of pure states~\cite{grier2022sample}. Variants of the protocol have been considered where the randomized measurements are: derandomized~\cite{Huang2021Derandomization}, reused~\cite{helsen2022thrifty, zhou2022performance} or generalized to joint measurements~\cite{grier2022sample}. Randomized measurements generated by low depth unitary ensembles~\cite{akhtar2022scalable,arienzo2022closedform}, locally-scrambled unitary ensembles~\cite{bu2022classical, Choi} or certain Hamiltonian evolutions~\cite{Hu2022} have also been considered.  

The HKP classical shadows protocol has two variants; one randomizes the computational basis measurement by first applying a global Clifford unitary sampled uniformly at random to each copy of the unknown quantum state (a.k.a.~random Clifford measurements), the other by applying a layer of random single qubit Clifford gates (a.k.a.~random Pauli measurements). Both prescriptions give rise to highly interesting and complementary schemes, 
each able to estimate expectation values for a large class of observables using a sparing number of copies of the unknown state. 
The global Clifford scheme performs well for observables with controlled Frobenius norm including quantum states (for fidelity estimation) and other low rank operators with bounded spectral norm. However, this scheme performs very poorly for high rank observables such as Pauli strings, even those with low weight. 
Additionally, the implementation of global Cliffords requires a linear depth circuit making the scheme infeasible even for moderate sized systems. In contrast, the single qubit Clifford scheme is technologically far less demanding and has seen experimental implementation~\cite{struchalin2021experimental,elben2020mixed,zhang2021experimental}, however, it requires a number of samples that scales exponentially in the weight of the observable restricting its application to observables supported on a small number of qubits. 
To make the powerful global Clifford scheme
into a tool useful in present-day laboratories and simultaneously endow it with the advantages of the single qubit Clifford scheme on local observables, it is key to identify and study a family of shadow schemes connecting these extremes.

A natural approach is to use randomized Clifford circuits modulated by circuit depth where the single qubit and global Clifford schemes of HKP can be seen as the zero and infinite depth extremes, respectively. 
However, intermediate schemes possess significant qualitative differences related to the fact that low depth Clifford circuits do not form a group. 
For intermediate schemes, two key technical challenges are:
\begin{enumerate}
\item  Inverting the channel associated with randomized measurements (a.k.a.~the  \emph{measurement channel}) to permit the estimation of expectation values from shadows. 
\item Bounding the \emph{shadow norms}, which in turn bound the variance of the relevant estimators, to obtain guarantees on the accuracy of estimates. 
\end{enumerate}
These are closely related to sufficiently characterizing the second and third moment operators with respect to the (depth-modulated) Clifford ensembles, respectively.
Accomplishing this is significantly more challenging for intermediate depths. For example, in the single qubit and global Clifford schemes of HKP, the measurement  channel has a high degree of symmetry\footnote{In the single qubit Clifford case, the measurement channel commutes with the canonical representation of the permutation group $S_n$ and the Pauli group. In the global Clifford case, it commutes with any global Clifford.}, and it is simply the local and global depolarizing channel, respectively, making inversion simple.
More generally, this channel is depth dependent and, in the intermediate regime away from these extremes, it loses some of its highly symmetric structure. 
A path to a practical, depth-modulated classical shadows scheme must overcome these obstacles.

Here, we present a practical, depth-modulated classical shadows scheme for low depths. We give methods for computing rigorous guarantees on the scheme's performance and discuss how the choice of depth and observable affects performance. More precisely, our protocol randomizes computational basis measurements using Clifford circuits of depth $d\in \{\infty,0,1,2,\ldots\}$ generated by geometrically local, two-qubit Clifford gates acting on $n$ qubits. 

We then present numerical methods for either exactly computing or upper bounding the depth-dependent shadow norm associated with a given observable, permitting guarantees on the trade-off between the number of copies of the state used and the accuracy of the estimates.
A key to our method is the probabilistic interpretation of the eigenvalues of the measurement channels and coefficients relevant in the computation of the shadow norm. This permits use of tensor networks to, efficiently in $n$, describe and manipulate key objects in the framework of classical shadows. 

We consider two classes of observables, both supporting an efficient description. These are; 
observables given as a linear combination of at most $\mathrm{poly}(n)$ Paulis (we call these \emph{sparse observables}) 
and observables given as a \emph{matrix product 
operator} (MPO) \cite{MPORep} with a $\mathrm{poly}(n)$ bounded 
bond dimension (we call these \emph{shallow observables}). An example of the former is local Hamiltonians, 
an example of the latter is projectors onto 
\emph{matrix product states} (MPS) \cite{MPSReps} for fidelity estimation. In the $d\leq \order{\log n}$ regime, our scheme permits efficient estimation of expectation values and efficient computation of shadow norm bounds with respect to sparse observables. In the shallow observable input model, we require a $\mathrm{poly}(n)$ bond dimension matrix product state representation of the inverse measurement channel. Given this, efficient estimation of expectation values and efficient computation of shadow norm bounds is achieved. We employ a variational method for computing an approximate inverse to the measurement channel produced in the form of a low bond dimension MPS. This approximation error is bounded, permitting rigorous bounds on the errors associated with expectation value estimation. In practice, our variational inversion procedure only needs to be executed once for any fixed depth $d$ and system size $n$ and works well for practically motivated parameter scales. 
However, for an arbitrary choice of parameters $n$ and $d$, this procedure is not guaranteed to produce a 
high accuracy approximate inverse within a run-time that is feasible or efficient in system size.

Our research was performed concurrently to that of Ref.~\cite{akhtar2022scalable} and their protocol overlaps significantly with ours. Nevertheless, our analysis goes further in some respects.
 Refs.~\cite{Choi,bu2022classical,akhtar2022scalable} focus on the \emph{entanglement features} formalism and the so called \emph{locally scrambled shadow norm} which bounds the sample complexity for typical input states\footnote{The locally scrambled shadow norm bounds the sample complexity averaged over all states in a one design, hence if it is bounded by some $B>0$, 
 for some $\delta>0$ a fraction $1-\delta$ of these states will have a sample complexity upper bounded by $\sim B/\delta$.} (as opposed to worst case input states). This quantity is upper bounded by the shadow norm but is used as a proxy for the shadow norm since it is easier to characterize mathematically and, for certain observables, is expected to be proportional to the shadow norm based on numerical evidence. We present bounds for the \emph{shadow norm}. Additionally, we present rigorous \emph{analytical} upper bounds to the locally scrambled shadow norm in the regime $d=\Theta(\log(n))$. This provides evidence that this regime achieves an ideal middle ground between the $d=0$ and $d=\infty$ cases in the following form. We show that for typical input states our $d=\Theta(\log(n))$ protocol is expected to yield, up to constant factors, the same sample complexity as the much more experimentally demanding global Cliffords scheme. 
 Additionally, we show that for the important class of bounded observables that are sparse and 
 have low-weight Pauli components
 (e.g., 1D local Hamiltonians), in contrast to the global scheme, our $d=\Theta(\log(n))$ scheme is provably sample efficient.

\section{Notation and overview} 

Throughout, we will use $\mathbb{I}$, $X$, $Y$, and $Z$ to denote the single qubit Pauli operators, unless otherwise stated. We use $\pauli{n}:=\set{\mathbb{I},X,Y,Z}^{\otimes n}$ to denote the $n$-qubit set of Pauli strings (without phases). We index the element of $\pauli{n}$ by $\lambda=(x,z)\in \bitstring{n} \times \bitstring{n}$ as follows
\begin{align}\label{eq:pauli_index}
    P^{\lambda}=\bigotimes_{j\in [n]} i^{x_j z_j} X^{x_j} Z^{z_j},
\end{align}
and use $\pm \mathcal{Z}$ to denote $\set{\pm P^{(0^n,z)}~|~z\in \bitstring{n}}$ i.e., the set of Pauli strings consisting only of $\mathbb{I}$ and $Z$ factors up to a phase factor of $\pm 1$.
We use $\clif{n}$ to denote the $n$-qubit Clifford group and $\oplus$ to denote binary vector addition modulo $2$.

In \cref{sec:hkp}, we give a brief summary of the method proposed in Ref.~\cite{Huang2020}. In \cref{sec:method}, we sketch out our depth-modulated classical shadows method. In the subsequent sections we then go into more details regarding the ingredients of our scheme. In Section \ref{sec:representations_M}, we discuss the properties of the \emph{measurement channel}, a key object in classical shadows, and construct efficiently computable expressions of its eigenvalues in term of probabilities. In section \ref{sec:inverse} we go over a method to obtain the inverse of the measurement channel, which is essentuial to compute expectation values. In \cref{sec:computing_shadows}, we discuss how to acquire data and manipulate it classically into an estimation of an expectation value within our scheme. Finally in \cref{sec:performance_guarantees}, we analyse the sample complexity of our method and provide numerical experiments in \cref{sec:numerics}.

\section{The HKP classical shadows protocol}\label{sec:hkp}
Given many identical copies of an unknown $n$-qubit quantum state $\rho$, the goal is to compute a classical description $\hat{\rho}$ such that for any observable $O$ in a large class, the classical description can be used to generate sufficiently accurate expectation values. Ref.~\cite{Huang2020} defines an operator norm $\shadownorm{O}$ and a procedure for computing a classical description $\hat{\rho}$ (the shadow) from $N$ copies of the unknown state $\rho$. In particular, the $N$-copy shadow can be though of as a list of $N$ independent copies each of which can be represented as a stabilizer state acted upon by the inverse of the measurement channel. After the shadow has been constructed, we are given a set of $M$ distinct observables $\set{O_i}_{i=1}^M$ with spectrum in the interval $[-1,1]$ and bounded norm, $\shadownorm{O_i}^2\leq B$. To estimate $e_i:=\tr{O_i \rho}$, the \emph{median of means} method 
\cite{lugosi2019mean}
is employed~\cite{lerasle2019lecture, Huang2020}. 

Here, the $1$-copy shadows are partitioned into equal sized blocks and averaged within each block to produce $\hat{\rho}^{(j)}$ for each block $j$. The estimate $\hat{e}_i$ is given by taking the median of $\tr{O_i \hat{\rho}^{(j)}}$ values over all blocks.
With this scheme, HKP show that a surprisingly small number of copies of $\rho$ suffice to construct $\hat{\rho}$ to sufficient accuracy such that for all $M$ observables the expectation values with respect to $\rho$ can be efficiently estimated from $\hat{\rho}$. More precisely, they show that with high probability over the randomness of the construction of the shadow, $\abs{e_i -\hat{e}_i }\leq \epsilon$ for all $i\in [M]$ provided $N\geq \Omega(B \log(M) \epsilon^{-2})$.

To construct the shadow, 
for each copy of $\rho$, a unitary $U$ is independently chosen from a subset $ \eset \subseteq \clif{n}$ of Clifford unitaries with probability $\eprob(U)$. This defines an ensemble of $n$-qubits Clifford gates $\ensemble=\defensemble$. From $\rho$, the state $U \rho U^{\dagger}$ is generated and measured in the computational basis to produce some computational basis state $\ketbra{b}{b}$. Undoing the unitary $U$ produces the stabilizer state $\sigma=U^{\dagger}\ketbra{b}{b}U$, also known as a \emph{snapshot} of $\rho$. The expected outcome of this process (over the ensemble and measurement outcomes) can be viewed as a linear map, which we will refer to as the \emph{measurement channel}.
\begin{align}\label{eq:measurement_map}
    \measmap{\rho}:=\expval{U\sim \ensemble}{\sum_{b \in \bitstring{n}} \matrixel{b}{U \rho U^{\dagger}}{b} U^{\dagger}\ketbra{b}{b}U}.
\end{align}
Assuming this channel, $\measmap{\cdot}$ is invertible, the operator $\invmeasmap{U^{\dagger} \ketbra{b}{b} U}$ defines a \emph{1-copy shadow}.
%with the $N$-copy shadow computed by averaging $N$ of these. 
The shadow (and hence also an average of independent shadows) is an unbiased estimators of $\rho$ i.e., it reproduces $\rho$ on expectation. Since expectation values are linear functions of $\rho$, a multi-copy average of shadows can be used as an unbiased estimator of expectation values; however, how quickly (in the number of copies) this estimate converges to the expectation value is determined by the variance of the estimator which is upper bounded by the squared shadow norm $\shadownorm{O}^2$,
\begin{align}\label{eq:shadow_norm}
    \underset{\sigma}{\max} \left\{ \expval{U\sim \ensemble}{\sum_{b \in \bitstring{n}}
    \matrixel{b}{U \sigma U^{\dagger}}{b} \matrixel{b}{U \invmeasmap{O} U^{\dagger}}{b}^2 }\right\},
\end{align}
where the maximization is over density matrices $\sigma$.

We note that the procedure for constructing the shadow, the measurement channel and the shadow norm all depend on the choice of ensemble $\ensemble$. Ref.~\cite{Huang2020} explicitly considers two ensembles we label as $\sensemble{0}=\defsensemble{0}$ and $\sensemble{\infty}=\defsensemble{\infty}$ where $\seset{0}$ consists of all $n$-fold tensor products of single qubit Clifford unitaries and $\seset{\infty}$ consists of all $n$-qubit Clifford unitaries with $\seprob{0}$ and $\seprob{\infty}$ being the uniform distribution over $\seset{0}$ and $\seset{\infty}$, respectively. Notice that linear depth circuits already suffice to implement any Clifford unitary exactly \cite{bravyiHadamardFree2021}, and hence this suffices to implement the HKP global Cliffords scheme. However, in the random circuit setting one needs $d=\infty$ to generate uniform randomness \textit{exactly}, hence the notation $\mathcal U_{\infty}$. %we $d=\infty$ as a notation because it is the depth the gate-wise random ensembles we introduce requires to generate uniform randomness \textit{exactly}. 
In these two extremal settings, the measurement channels (denoted $\smeasmap{\cdot}{0}$ and $\smeasmap{\cdot}{\infty}$) have a particularly simple form; acting as the product of single qubit depolarizing channels and a global depolarizing channel, respectively. That is, for $P\neq \mathbb{I}$ an $n$-qubit Pauli, 
these act as
\begin{align}\label{eq:HKP_meas_ops}
    \smeasmap{P}{0}=3^{-\wt{P}}P \quad \text{and} \quad \smeasmap{P}{\infty}=\frac{1}{2^n+1}P,
\end{align}
where $\wt{P}\in \set{0,1,\ldots,n}$ denotes the weight of the Pauli, 
i.e., the number of non-identity tensor factors. More generally, the action of the measurement channel can be defined by linear extension from \cref{eq:HKP_meas_ops}. Combined with the observation that $\smeasmap{\mathbb{I}^{\otimes n}}{0}=\mathbb{I}^{\otimes n}=\smeasmap{\mathbb{I}^{\otimes n}}{\infty}$, \cref{eq:HKP_meas_ops} also establishes that these measurement channels are invertible.

\section{Overview of the method}\label{sec:method}
In this section we give a step by step overview of how our method can be implemented.
\begin{tcolorbox}[title= Depth-modulated classical shadows]
\begin{enumerate}
    \item Given $N$ copies of some state $\rho$ on $n$ qubits, pick $d$. We expect the ideal $d$ to scale like $\log(n)$, both in terms of efficiency of implementation and in terms of sample efficiency for the largest possible class of observables.
    \item Construct the MPS representation of the measurement channel $\mathcal M_d$ as per \cref{sec:representations_M}.
   \item Sample the $N$ classical shadows as per \cref{sec:computing_shadows}.
    \item Choose a set of observables. If there are any shallow observables, obtain an MPS representation of $\mathcal M_d^{-1}$ as per \cref{sec:inverse}.
    \item Compute the corresponding expectation values as per \cref{sec:computing_shadows}.
    \item Upper bound the shadow norms of the chosen observables as per \cref{sec:performance_guarantees}. 
    \item Use standard techniques to bound the accuracies of estimates from upper bounds on the variances of the unbiased estimators.
\end{enumerate}
\end{tcolorbox}

\section{representations of the measurement channel}\label{sec:representations_M} We consider the extension of the classical shadows protocol to the intermediate Clifford depth regime. 
We define the ensemble of Clifford unitaries, $\sensemble{d}=\defsensemble{d}$, where for `circuit depth' $d\in \{\infty,0,1,\dots\}$, $\seset{d}$ denotes the set of allowed unitaries and $\seprob{d}$ denotes the measure over these. Concretely, we sample from $\sensemble{d}$ by first applying a uniformly random single qubit Clifford gate to each qubit (the $0^{\mathrm{th}}$ layer)\footnote{Strictly speaking, this first layer does not change the ensemble and can be omitted, we include it here as then the $d=0$ case of our ensemble corresponds to the random Pauli measurements scheme of Ref.~\cite{Huang2020}.} then for each layer $i\in [d]$ we apply uniformly randomly chosen two qubit Clifford gates in a nearest-neighbour circular brickwork pattern applied to $n$ (even) qubits in a circle, 
i.e., 
the first and last qubit are identified as neighbours (see Fig.~\ref{fig:brickwork}). With some modifications, our work generalizes to qubits on a line (non-circular), arbitrary $n$ and to other architectures but we do not consider these here for clarity and concreteness.

We note that (for $n$ even), the $d=0$ and $d=\infty$ limits of our ensembles $\sensemble{d}$ recover the $\sensemble{0}$ and $\sensemble{\infty}$ ensembles studied in Ref.~\cite{Huang2020}. 
In fact, for $d=\mathcal O(n)$, $\sensemble{d}$ becomes an approximate $3$-design~\cite{brandao2016local} and recovers the properties of $\sensemble{\infty}$ for our purposes. Using \cref{eq:measurement_map,eq:shadow_norm}, we define the measurement channel $\smeasmap{\cdot}{d}$ and $\sshadownorm{\cdot}{d}$ with respect to the ensemble $\sensemble{d}$. We show that this channel is diagonal in the Pauli basis and give a probabilistic expression for the eigenvalues.

\begin{figure}
    \centering
    \includegraphics[scale=.3, trim= 0 100 400 100]{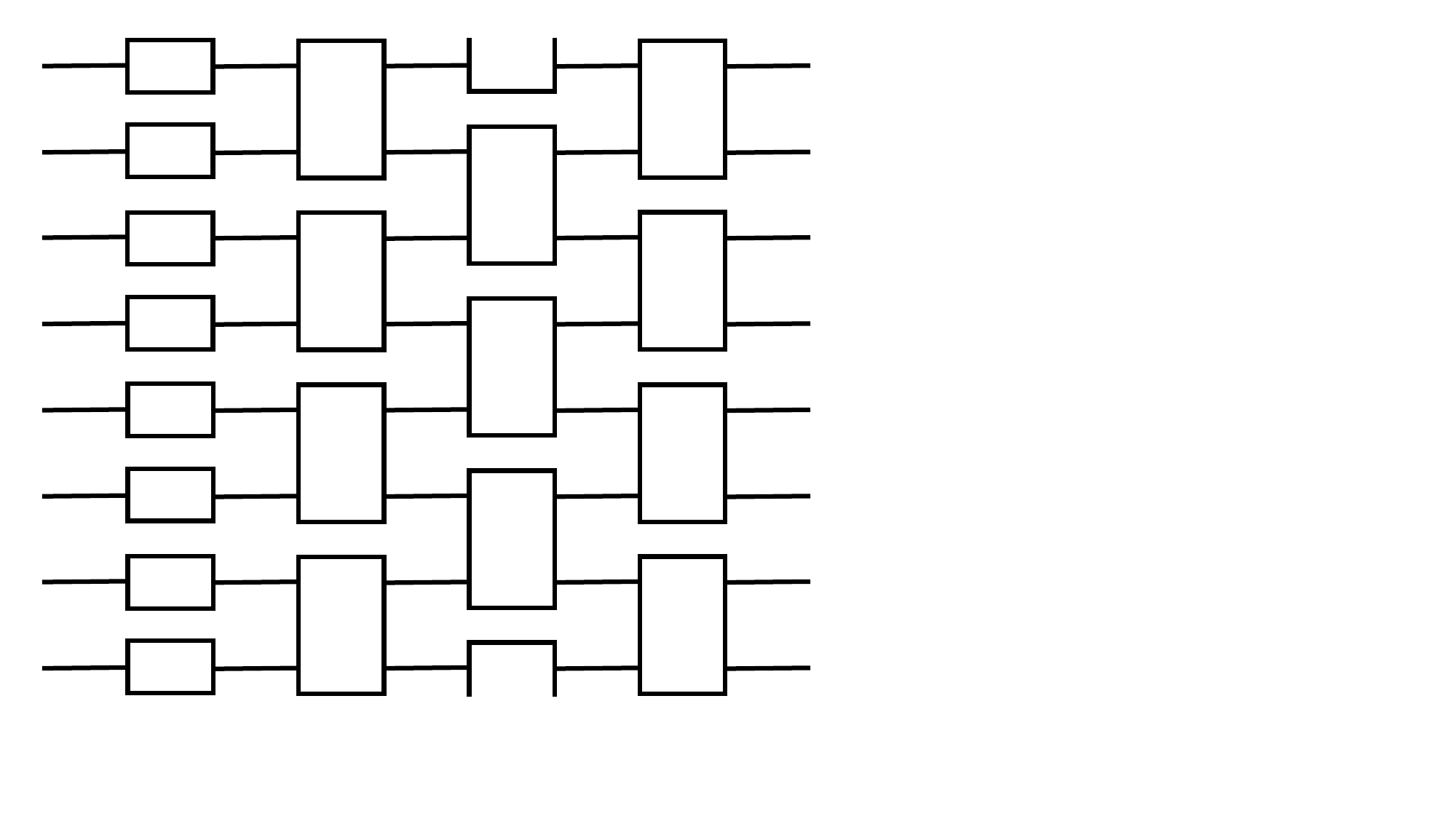}%Hakop's diagram
    %trim= 1 2 3 4 ‘crops’ the picture by 1bp at the left, 2bp at the bottom, 3bp on the right and 4bp at the top
    \caption{Nearest-neighbour circular brickwork architecture for $n=8$ qubits and depth $d=3$. Consists of $3$ layers (excluding) the $0^{\rm th}$ layer of single qubit gates. The second layer contains a two qubit gate acting on qubits $1$ and $8$ which are identified as nearest-neighbour in this architecture.}
    \label{fig:brickwork}
\end{figure}

\begin{lem}\label{lem:measurement_map_action}
For $\lambda \in \bitstring{2n}$, let $\tprobl{\lambda}\in [0,1]$ be 
defined as
    \begin{align}\label{eq:t as a prob}
        \tprobl{\lambda}:=\prover{U\sim \sensemble{d}}{UP^{\lambda}U^{\dagger}\in \pm\mathcal{Z}},
    \end{align}
then,
\begin{align}\label{eq:meas_diag_on_paulis}
    \smeasmap{P^{\lambda}}{d}=\tprobl{\lambda}P^{\lambda}.
\end{align}
\end{lem}

The proof of \cref{lem:measurement_map_action} can be found in \cref{app:proof_lem1}. We note that the diagonal action of the measurement channel in the Pauli basis was independently identified in Ref.~\cite{bu2022classical}. 
This probabilistic interpretation has an important conceptual value as it allows to interpret $t_{\lambda,d}$ as the probability of a certain outcome of a classical random walk on bit strings, rather than on the Clifford group. As a matter of fact since our ensemble is invariant under local Cliffords $t_{\lambda,d}$ only depends on the support of $P^{\lambda}$.
To recover 
Eq.~(\ref{eq:HKP_meas_ops}), corresponding to the two extremes analyzed in HKP, from \cref{lem:measurement_map_action}, we make the observation that $U\in \seset{0}$ must preserve Pauli weight, i.e., $\wt{P}=\wt{UPU^{\dagger}}$ sending each single qubit non-identity tensor factor of $P$ to $\pm Z$ with probability $1/3$, hence, $\stprobl{\lambda}{0}=3^{-\wt{P^{\lambda}}}$. Further, for $U\sim \sensemble{\infty}$ and $P^{\lambda}\neq \mathbb{I}^{\otimes n}$, the Pauli $UP^{\lambda}U^{\dagger}$ is equally likely to be any non-identity Pauli giving
\begin{align*}
    \stprobl{\lambda}{\infty}=\frac{2^n -1}{4^n -1}=\frac{1}{2^n +1}.
\end{align*}

We now present efficient methods for exactly computing both $\tprobl{\lambda}$ for a given $\lambda$ and $\smeasmap{\sigma}{d}$ for a state $\sigma$ given as a low\footnote{Meaning, upper bounded by $\mathrm{poly}(n)$.} bond dimension MPO provided that $d=\order{\log(n)}$. The way in which we resort
to tensor network techniques to arrive at scalable
randomized schemes for quantum systems identification
is reminiscent of that of Ref.~\cite{ohliger_efficient_2013}. 
We will work in the Pauli basis $\pauli{n}$ as it diagonalizes the measurement channel (c.f. \cref{eq:meas_diag_on_paulis}). 

Any operator $\sigma$ acting on $(\bbc^2)^n$ can be written in the Pauli basis as 
\begin{align}\label{eq:sigma_as_MPO}
    \sigma=\sum_{\lambda} \alpha_\lambda P^{\lambda},
\end{align}
with the coefficients $\alpha_\lambda$ in the form
\begin{align}\label{eq:MPS_form}
    \alpha_\lambda=\tr{C_{\lambda_1} C_{\lambda_2} \cdots C_{\lambda_n}},
\end{align}
where $C_{0,0}, \ldots, C_{1,1}$ are $b_{\alpha} \times b_{\alpha}$ matrices and $b_{\alpha}$ is known as the bond dimension of the vector $\alpha_{\bullet}$. We will be interested in operators $\sigma$ such that the bond dimension is at most polynomial in $n$. In general, a bond dimension that scales exponentially in $n$ is needed to express an arbitrary operator via \cref{eq:MPS_form}. \cref{eq:MPS_form} takes the form of a matrix product state. This form is particularly useful since contractions over exponentially many indices can be computed efficiently. For example,
\begin{align}\label{eq:MPS_contraction}
    \sum_{\lambda} \alpha_{\lambda}=\tr{\left(\sum_{\lambda_1} C_{\lambda_1}\right)\times \ldots \times \left(\sum_{\lambda_n} C_{\lambda_n}\right)}
\end{align}
with the right hand side being efficiently computable in $n$. By \cref{eq:meas_diag_on_paulis}, we have
\begin{align}
    \smeasmap{\sigma}{d}=\sum_{\lambda} \alpha_{\lambda} \tprobl{\lambda} P^{\lambda},
\end{align}
 We will write the vector $\tprobl{\bullet}$ as an MPS. We note that if the bond dimensions of $\alpha_{\bullet}$ and $\tprobl{\bullet}$ are $n^{\order{1}}$ then $\smeasmap{\sigma}{d}$ can be written as an MPO with $n^{\order{1}}$ bond dimension. In particular, the updated vector of Pauli coefficients $\alpha'_{\bullet}$ can be written as an MPS with bond dimension that is the product of the bond dimensions of $\alpha_{\bullet}$ and $\tprobl{\bullet}$ where $\alpha'_{\lambda}:=\alpha_{\lambda} \tprobl{\lambda}$. This will permit the efficient computations of expectation values with respect to $n^{\order{1}}$ bond 
dimension MPO observables.

\begin{lem}\label{lem:meas_map_comp}
$t_{\lambda,d}$ can be written as an MPS with bond dimension at most $2^{d-1}$.
\end{lem}

We present the proof of \cref{lem:meas_map_comp} in the form of an explicit construction in \cref{app:proof_lem2} where we show that $t_{\lambda,d}$ is given as the tensor element depicted in \cref{fig:tprob_TN_MPS}.
\begin{figure}
    \centering
    \includegraphics[scale=.3, trim= 0 100 550 100]{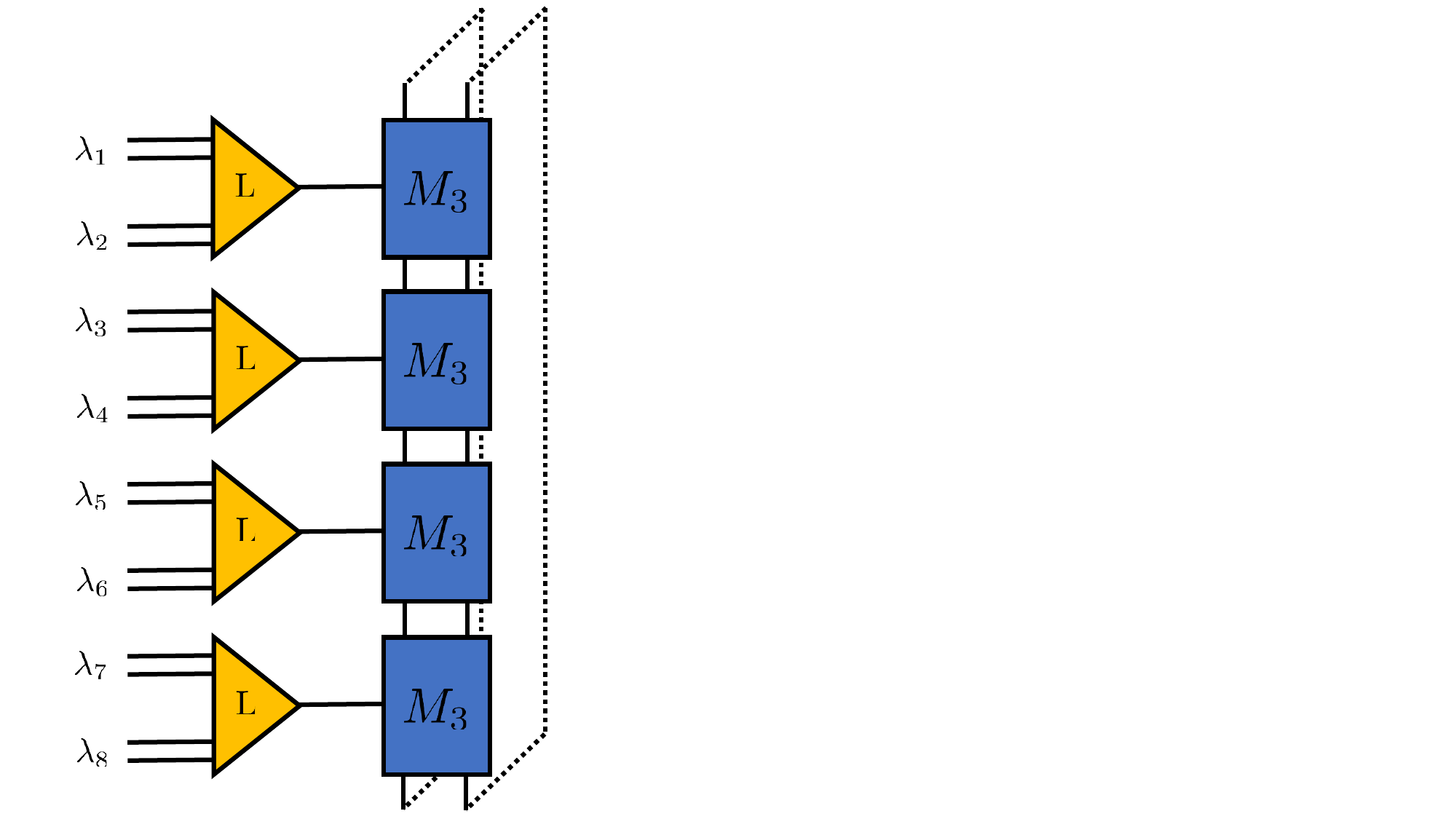}
    %trim= 1 2 3 4 ‘crops’ the picture by 1bp at the left, 2bp at the bottom, 3bp on the right and 4bp at the top
    \vspace{0.7cm}
    \caption{Tensor network for computing $t_{\lambda,d}$ for $n=8$ qubits and depth $d=3$. Each $L$ tensor is a logical operator taking the $\textsc{or}$ of two $\textsc{or}$ gates. The $M_d$ tensors can be constructed from the random Clifford architecture 
    (c.f.~\cref{fig:brickwork}) as per \cref{app:proof_lem2}. 
    When its horizontal leg input value $h\in \set{0,1}$ is assigned, the tensor becomes a transfer matrix $M_{d,h}$ of dimension $2^{d-1} \times 2^{d-1}$.}
    \label{fig:tprob_TN_MPS}
\end{figure}

The following follows immediately from \cref{lem:meas_map_comp}.
\begin{cor}\label{cor:transition_prob_comp}
For depth $d=\order{\log n}$, the probabilities $t_{\lambda,d}$ can be exactly computed in run-time $n^{\order{1}}$.
\end{cor}

For the shadow protocol to work, it is essential that the measurement channel be invertible. From \cref{eq:t as a prob}, we observe that for any $\lambda\in\{0,1\}^{2n}$, the probability $t_{\lambda,d}$ is non-zero ensuring non-singularity and that the inverse is given by
\begin{align}
    \invsmeasmap{P^{\lambda}}{d}=\frac{1}{\tprobl{\lambda}}P^{\lambda}.
\end{align}

\section{Computing shallow shadows and estimating expectation values}\label{sec:computing_shadows} We now specify how copies of the unknown state $\rho$ are measured and stored in a form that specifies the shadow $\hat{\rho}$. This step of our procedure is analogous to Ref.~\cite{Huang2020}. That is, we first sample $U\sim \sensemble{d}$ and apply $U$ to a copy of $\rho$. We then measure in the computational basis. The output computational basis state $\ketbra{b}{b}$ is then conjugated by $U^{\dagger}$ to produce a snapshot $\sigma=U^{\dagger}|b\rangle\langle b|U$ that depends on the random variables $U$ and $b$.
The rest of the procedure is based on the fact that the expected value of $\sigma$ is by definition $\mathcal M_d(\rho)$ (cf. \cref{eq:measurement_map}), hence the expectation value of $\invsmeasmap{\sigma}{d}$ is $\rho$. $\invsmeasmap{\sigma}{d}$ is then an unbiased estimator of $\rho$. As the trace is linear, for any observable $O$, we then get that $\tr{O\invsmeasmap{\sigma}{d}}$ is an unbiased estimator of the expectation value $\tr{O\rho}$. 
We note here that since $\mathcal M^{-1}_d$ is self adjoint, $\tr{\invsmeasmap{O}{d}\sigma}$ is an unbiased estimator of $\tr{O\rho}$ as well. To estimate $\tr{\rho O}$ we then have to efficiently sample and store a large number of instances of this estimator. The sample efficiency of the estimation depends on the second moment of the estimator. This is upper bounded by the shadow norm, as we discuss in \cref{sec:performance_guarantees}.
At this point the procedure is different depending on the type of observables. If $O$ is a \emph{sparse} observable, we have access to an efficient description of $O$ in terms of Paulis, meaning that we have
\begin{equation}
    O=\sum_{k=1}^r \beta_{\lambda_k} P^{\lambda_k}
\end{equation}
with $r\leq n^{\order{1}}$. We can then efficiently compute $t_{\lambda_k,d}$ for each $k$ and simply store the inverted observable
\begin{equation}
    \invsmeasmap{O}{d}=\sum_{k=1}^r \frac{\beta_{\lambda_k}}{t_{\lambda_k,d}} P^{\lambda_k}.
\end{equation}
 Computing the expectation value is now a matter of efficiently computing $\langle b|U P^{\lambda_k}U^{\dagger}|b\rangle$ for every snapshot, which can be done efficiently as $U|b\rangle$ is a stabilizer state.

If instead $O$ is a \emph{shallow} observable, it is given as
\begin{equation}
    O=\sum_{\lambda}\beta_{\lambda}P_{\lambda}
\end{equation}
where $\beta_{\lambda}$ is representable as an MPS with low bond dimension, or equivalently $O$ might be represented as an MPO in any local basis, since local basis changes do not affect the bond dimension. For example $O$ might be a rank one projector corresponding to an MPS state.
In this case, it is necessary to construct an MPS representation for $\invsmeasmap{\bullet}{d}$ as described in \cref{sec:inverse}. We note that since $\ket{b}$ is a product state vector and $U\in \seset{d}$ is depth $d$, the stabilizer state $\sigma$ can be written as an MPO (c.f. \cref{eq:sigma_as_MPO}) with bond dimension $b_{\alpha}\leq 4^{d-1}$.
Depending on the relative number of observables to estimate and snapshots stored, it might be more efficient to use tensor contractions to store either $\invsmeasmap{\sigma}{d}$ for every snapshot $\sigma$ or $\invsmeasmap{O}{d}$ for every observable $O$ we are interested in. In any case, for every snapshot $\sigma$ we can then efficiently compute $\tr{O \invsmeasmap{\sigma}{d}}$ by contracting the tensor networks corresponding to $O$, $\sigma$, and $\invsmeasmap{\bullet}{d}$.
\section{Heralded matrix product operator inversion}\label{sec:inverse} 
As explained in the previous section, in order to estimate shallow observables, it is necessary to obtain an MPS representation of the eigenvalues of the inverse measurement map.
We present a variational DMRG-like method inspired by Ref.~\cite{guo2022quantum} for computing a set of matrices $\{V\}=\{V_0^j, V_1^j\}_{j=1}^{N}$ parametrising an MPS that approximates $1/\tprobl{\lambda}$ on all 
inputs $\lambda$. More generally, given MPS tensors $\{M_0^j, M_1^j\}_{j=1}^{N}$ we want to find new tensors $\{V_0^j, V_1^j\}_{j=1}^{N}$ such that $m(x)v(x)\approx 1$ for all $x\in\{0,1\}^{N}$, where we define
\begin{equation}
    m(x)=\tr{M_{x_1}^1M_{x_2}^1\dots M_{x_{N}}^{N}}.
\end{equation}

And similarly for $v(x)$. Our procedure works by minimizing the following cost function, 
\begin{eqnarray}\label{eq: cost_function}
    C(\{V\})^2&:=&\sum_{x\in \{0,1\}^{N}} \abs{m(x)v(x)-1}^2
    \\
     \nonumber
    &=&\sum_{x\in \{0,1\}^{N}} (\abs{m(x)v(x)}^2-2 \mathrm{Re}(m(x)v(x))+1)
\end{eqnarray}
which can be efficiently computed using \cref{eq:MPS_contraction}. We minimize this cost function using a form of \emph{cyclic coordinate descent} (cf. \cref{alg:MPSinv}), that is, starting from a random ansatz, we optimize the cost function by varying one of the matrices in $\{V\}$, while keeping the others fixed, and repeat the process with all the matrices in $\{V\}$. Sweeping over all the matrices in this way the algorithm converges to a minimum.
\begin{algorithm}[H]
\caption{Algorithm for computing an approximate inverse MPS}\label{alg:MPSinv}
\begin{algorithmic}[1]
\algrenewcommand\algorithmicrequire{\textbf{Input:}}
\algrenewcommand\algorithmicensure{\textbf{Output:}}
\Require MPS $\{M_0^j,M_1^j\}_{j=1}^{N}$, ansatz $\{V^j_0,V^j_1\}_{j=1}^{N}$, threshold $\epsilon$.
\Ensure Approximate inverse $\{V^j_0,V^j_1\}_{j=1}^{N}$
\Statex
\While{$C(\{V\})>\epsilon$}
\For{$j=1,\dots N$}
\For{$k=0,1$}
\State $V^j_k \xleftarrow{} \underset{V_{k}^j}{\mathrm{argmin}} \, C(\{V\})$
\EndFor
\EndFor
\EndWhile
\State{\Return $\{V^j_0,V^j_1\}_{j=1}^{N}$}
\end{algorithmic}
\end{algorithm}

The local optimization, i.e. step 4 in \cref{alg:MPSinv}, is particularly simple, as the cost function when all matrices except one are fixed is simply a bounded quadratic form, this implies that the minimum at each optimization step is unique, and hence the procedure is guaranteed to converge to a stationary point of the cost function \cite{GRIPPO2000127}. A detailed description of the algorithm can be found in Appendix \ref{app:variational_inverse}. Our procedure is heralded in the sense that if the procedure finds a $\{V\}$ with a small cost function, then we are guaranteed that this can be used to produce expectation value estimates with negligible error arising from this approximate inverse procedure. In fact, while the cost function in \cref{eq: cost_function} is easy to compute, it is expected that the actual error given by the approximation should be smaller, since it scales like $\max_x|m(x)v(x)-1|$ instead. We prove the following in Appendix \ref{app:proof_approx_m}:
\begin{lem}\label{lem:approx_inverse}
    Let $\mathcal{V}$ be a channel that is diagonal in the Pauli basis with coefficients $v_{\lambda}$ for each $P^{\lambda}\in \pauli{n}$ and is an approximate inverse of the measurement channel in the sense that there exists $\epsilon>0$ such that for all $\lambda\in \bitstring{2n}$,
    $\abs{1-\tprobl{\lambda}v_{\lambda}}\leq \epsilon$. Then
        \begin{equation}
        |\tilde o-\tr{O\rho}|\leq \epsilon||O||_{F}
    \end{equation}
    where $\tilde o$ is the expectation value obtained by using the approximate inverse channel in the shadow protocol.
\end{lem}
Hence the bias induced by the approximate inverse is small for low-Frobenius norm observables. We will argue in the next section that, in the case of shallow observables where this inversion is necessary, these are also the observables for which the protocol is sample-efficient.
Finally, we note that for a given $n$ and $d$, the computation of a high accuracy variational inverse is only needed once. Hence, its run-time is a one-off cost for each new $n$, $d$ pair.

\section{Performance guarantees}\label{sec:performance_guarantees}
Shadow estimation is useful only insofar as the number of samples one needs to get to obtain a good estimate is well behaved, and in practice one needs to be able to estimate the number of samples needed to achieve a certain accuracy. For these purposes, the main quantity of interest is the shadow norm. For a given observable $O$, its shadow norm upper bounds the second moment (and hence also the variance) of the estimator $\tr{\hat{\rho}O}$ where $\hat{\rho}$ is a random 1-copy shadow. 
%(c.f.~\cref{eq:snapshot}). 
In this section, unless otherwise stated, we will always assume $d\leq \order{\log n}$. We first present methods for computing upper bounds on the shadow norm of various classes of observables. In particular we show how to exactly and efficiently compute the shadow norm of a Pauli operator, how to efficiently compute upper bounds on the shadow norm of a sparse operator and, how to efficiently compute upper bounds on the shadow norm of a shallow observable given an approximate inverse to the measurement channel in the form of a low bond dimension MPS.

We also consider the important sub-set of sparse bounded operators that can be written as a linear combination of $\order{\log n}$ weight Pauli operators such as, e.g., local Hamiltonians. We show a key result that all such observables have an $n^{\order{1}}$ shadow norm and hence also an $n^{\order{1}}$ sample complexity when $d=\Theta(\log n)$. We go on to give further evidence that the $d=\Theta(\log n)$ regime achieves the ``sweet spot'' between the $d=0$ and $d=\infty$ schemes of HKP in the sense that it permits the broadest class of observables that admit an $n^{\order{1}}$ shadow norm.

Let us start by defining the \emph{state dependent shadow norm} $\ssshadownorm{O}{d}{\rho}^2$
\begin{align*}\label{eq:shadow_norm}
  \expval{U\sim \sensemble{d}}{\sum_{b \in \bitstring{n}}
    \matrixel{b}{U \rho U^{\dagger}}{b} \matrixel{b}{U \invmeasmap{O} U^{\dagger}}{b}^2 }.
\end{align*}
This is simply the second moment of the estimator of $\tr{\rho O}$. Since $\rho$ is unknown, the  shadow norm $\sshadownorm{O}{d}$ is obtained by maximizing over all states $\rho$, i.e., by considering the worst possible variance. We can also define the \emph{locally scrambled shadow norm} \cite{bu2022classical,Choi} as the average, rather than the maximization, over all states in a state one-design, i.e., any ensemble of states $\mathcal E$ such that $\expval{\sigma\sim \mathcal E}{\sigma}=\mathbb{I}/2^{n}$,
as
\begin{equation}
    \ssshadownorm{O}{d}{\mathrm{LS}}^2:= \expval{\sigma\sim \mathcal E}{\ssshadownorm{O}{d}{\sigma}^2}.
\end{equation}
This quantity can be interpreted as quantifying the average performance of the scheme, i.e., the performance for a ``typical state". We show that all of these objects are actually norms in Appendix \ref{app:performance_guarantees}.
The goal is now to present methods to efficiently compute upper bounds for the shadow norm. The following quantity, a higher order analogue to $\stprobl{\lambda}{d}$, will be central for this task. For $\lambda, \lambda'\in \bitstring{2n}$, we define
    \begin{align}\label{eq:tau_as_a_prob}
        \ttprobl{\lambda,\lambda'}:=\prover{U\sim \sensemble{d}}{UP^{\lambda}U^{\dagger}, UP^{\lambda'}U^{\dagger}\in \pm\mathcal Z}.
    \end{align}
That is, $\ttprobl{\lambda,\lambda'}$ is the probability that $U$ sampled from the ensemble $\sensemble{d}$ will map both Paulis $P^{\lambda}$ and $P^{\lambda'}$ to a tensor product of $\mathbb{I}$ and $Z$ operators up to a factor $\pm 1$. We present some useful properties of $\tau_{(\lambda,\lambda')}$ in  \cref{app:performance_guarantees}, \cref{lem:ttprob_properties}. 
\begin{lem}
$\ttprobl{\lambda,\lambda'}$  can be written as an MPS with bond dimension $2^{\order{d}}$
\end{lem}
The proof is presented in \cref{app:proof_ttprob_MPS}, the next result follows immediately:

\begin{cor}\label{lem:ttprob_comp}
For depth $d=\order{\log n}$, the probabilities $\ttprobl{\lambda,\lambda'}$ can be exactly computed in run-time $n^{\order{1}}$.
\end{cor}
We now give a useful expression for the state dependent shadow norm:
\begin{thm}\label{lem:state_dep_snorm}
    For any observable $O$ with 
    Pauli decomposition $O=\sum_{\lambda}\beta_{\lambda}P^{\lambda}$ and any state $\sigma$, the square of the state dependent shadow norm is given by
    \begin{equation}\label{eq:s_norm_expression}
    \begin{aligned}
        \ssshadownorm{O}{d}{\sigma}^2&=\tr{\sigma\sum_{\lambda,\lambda'}\beta_{\lambda}\beta_{\lambda'}\frac{\ttprobl{\lambda,\lambda'}}{\tprobl{\lambda} \tprobl{\lambda'}} P^{\lambda}P^{\lambda'}}\\&=||O||_{s(d),\mathrm{LS}}^2+\tr{\sigma \tilde O}
        \end{aligned}
    \end{equation}
    where
     \begin{equation}
        \tilde O=\sum_{\substack{\lambda', \lambda \in \bitstring{2n}\\\lambda'\neq\lambda}}P^{\lambda}P^{\lambda'} \beta_{\lambda} \beta_{\lambda'} \frac{\ttprobl{\lambda,\lambda'}}{\tprobl{\lambda} \tprobl{\lambda'}}.
    \end{equation}
   Furthermore,
    \begin{equation}
        ||O||_{s(d),\mathrm{LS}}^2=\sum_{\lambda\in \bitstring{2n}} \frac{\beta_{\lambda}^2}{\tprobl{\lambda}}.
    \end{equation}
\end{thm}

We present the proof in \cref{app:performance_guarantees}. We isolate the identity component of $\sigma$ from the rest, as this gives rise to the locally scrambled shadow norm and in the $d=\infty$ case this is the dominant term, but as a simple corollary we may note that $||O||_{s(d)}^2$ is given by the largest eigenvalue of the matrix
\begin{equation}
    \sum_{\lambda,\lambda'}\beta_{\lambda}\beta_{\lambda'}\frac{\ttprobl{\lambda,\lambda'}}{\tprobl{\lambda} \tprobl{\lambda'}} P^{\lambda}P^{\lambda'}
\end{equation}
and the non locally-scrambled part by the largest eigenvalue of $\tilde O$.
 Given access to low bond dimension MPS representations of $\beta_{\bullet}$ and $1/t_{\bullet,d}$, the locally scrambled shadow norm can be efficiently computed exactly. 
We can immediately see that if the observable is a single Pauli, the second term $\tr{\sigma \tilde O}$ vanishes and the shadow norm corresponds to the inverse of the Pauli's eigenvalue, as noted in  \cite{bu2022classical}.
\begin{cor}\label{cor:pauli_snorm}

For $\lambda \in \bitstring{2n}$, the squared shadow norm of $P^{\lambda}$ is
\begin{align*}
    \sshadownorm{P^\lambda}{d}^2=||P^{\lambda}||_{s(d),\mathrm{LS}}^2=\frac{1}{\tprobl{\lambda}}.
\end{align*}
\end{cor}
Hence, for $d=\mathcal O(\log n)$ and $d=\infty$, all Pauli shadow norms can be efficiently calculated using \cref{cor:transition_prob_comp} and \cref{cor:pauli_snorm}. In this regime we can rigorously upper bound the shadow norm of certain linear combinations of Paulis. The following is a simple consequence of the triangle inequality, Corollary \ref{cor:pauli_snorm}, and \cref{thm:av_shadow_norm_UB}, which we will present at the end of this section:
\begin{thm}\label{lem:local_Hamiltonian_snorm}
Let $O=\sum_{k=1}^r\beta_{\lambda_k}P^{\lambda_k}$ be a linear combination of Pauli operators such that each Pauli is supported on a region of length upper bounded by $\mathcal O(\log(n))$, then if $d\leq \mathcal O(\log(n))$
\begin{equation}
    \sshadownorm{O}{d}\leq n^{\order{1}} \sum_{k=1}^r |\beta_{\lambda_k}|.
\end{equation}
\end{thm}
Beyond single Pauli operators, an important case in which the above is useful is local Hamiltonians, for which $\sum_{k=1}^r |\beta_{\lambda_k}|\leq n^{\order{1}}$, we then find that at logarithmic depth polynomially many samples are still sufficient to estimate local Hamiltonians to arbitrary precision, a property that is lost at $d=\infty$. For this class of observables, we note that both the estimation of expectation given a shadow and the computation of a polynomial upper bound on the shadow norm can be computed efficiently and without the use of an approximate variational inverse.

More generally, upper bounds on the second term in \cref{eq:s_norm_expression} may be useful.
\begin{lem}\label{lem:shadow_norm_bounds}
    For any observable $O$ and depth $d$,
    \begin{equation}\label{eq:sn_frob_bound}
        ||O||_{s(d)}^2\leq ||O||_{s(d),\mathrm{LS}}^2+||\tilde O||_F
    \end{equation}
    Furthermore, for $d=\mathcal O(\log(n))$, the right hand side of \cref{eq:sn_frob_bound} can be computed in $n^{\mathcal O(1)}$ time via tensor contraction given a low bond dimension MPS representation of $1/t_{\lambda,d}$.
\end{lem}

The bound in \cref{eq:sn_frob_bound} is obtained by upper bounding $\tr{\sigma \tilde O}$ in \cref{eq:s_norm_expression} with $||\tilde O||_F$, while this might look loose, it is a good bound for $d= \infty$ in the sense that $||\tilde O||_F\sim ||O||_{\mathrm{LS}}^2$, hence this only worsens the locally-scrambled shadow norm by a constant fraction, for moderately high depths we might then expect this relaxation not to worsen the variance estimation dramatically.
We give a proof of the second part of the lemma in \cref{app:performance_guarantees} by explicitely constructing the corresponding tensor contraction.
It is worth pointing out that in the MPS used to compute the right hand side of \cref{eq:sn_frob_bound} the bond dimensions of the various components ($\ttprobl{\bullet}$, $\beta_{\bullet}$ and $1/t_{\bullet,d}$, respectively) conspire to render the contraction infeasible beyond vary small depths. We expect a significant reduction to this bond dimension is possible but leave this to future work. Finally, we consider the shadow norm \emph{stabilizer projector}, that is it is the projector onto the space stabilized by an abelian subgroup of the Pauli group generated by $k\leq n$ elements. 
\begin{thm}\label{thm:stab_state_snorm}
    Let $S$ be the subspace stabilized by $\mathcal S=\langle s_1,\dots, s_k\rangle$, where $s_i\neq -I$ are commuting Pauli operators. Let $\Pi$ be the projector onto $S$. Then
    \begin{equation}\label{eq:stab_state_snorm}
        ||\Pi||_{s(d)}^2=||\Pi||_{s(d),|\psi\rangle \langle\psi|}^2=4^{-k}\sum_{\lambda,\lambda'\in \Lambda_\mathcal S}\frac{\ttprobl{\lambda,\lambda'}}{\tprobl{\lambda}\tprobl{\lambda'}}.
    \end{equation}
    for any $|\psi\rangle\in S$, where $\Lambda_{\mathcal S}=\{\lambda\in \{0,1\}^{2n}| P^{\lambda}\in \mathcal S\}$.
\end{thm}
The proof can be found in \cref{app:stab_state_snorm}. Notice that, given a low bond dimension MPS for $1/\tprobl{\bullet}$, the last expression in \cref{eq:stab_state_snorm} is efficiently computable if in addition $\Pi$ is given by a low bond dimension MPO, that is, there is a low bond dimension MPS deciding whether a Pauli is a stabilizer or not. An example is when $n=k$ and the corresponding stabilizer state is a low bond dimension MPS, e.g. the GHZ state.

In addition to being able to compute the shadow norm for specific observables, it would be desirable to be able to determine whether the shallow shadows protocol is efficient for a determinate observable based on coarser properties, and in general, whether it is reasonable to expect that our protocol is useful in the regime in which it is efficient, i.e., logarithmic depth. In Ref.~\cite{Huang2020}, the shadow norm of general observables was computed for the cases $d=0$ and $d=\infty$. In particular, it was found that for $d=\infty$, the shadow norm of $O$ is proportional to its Frobenius norm, which makes it ideal for estimating fidelities of quantum states or more generally expectation values of low rank observables with bounded operator norm, but not for local observables or high rank operators such as Pauli strings, which have exponentially large Frobenius norm. By contrast, the shadow norm of any local operator is well behaved for $d=0$. We provide evidence for and conjecture that a logarithmic depth circuit is the ideal middle ground between these two extremes, capturing enough of the long range entanglement to estimate global observables, while not losing so much local information. To fully prove this claim, one would need to analytically compute the shadow norm of a general observable, which, being a third moment of the circuit, is challenging. As evidence for this claim, we instead present rigorous \emph{analytic} bounds on the locally scrambled shadow norm. This quantity is upper bounded by the shadow norm, $||O||_{s(d),\mathrm{LS}}\leq ||O||_{s(d)}$, nevertheless in the extreme case $d=\infty$ it exhibits the same scaling \cite{Huang2020}. While not providing a worst case performance guarantee, we can interpret the value of $ ||O||_{s(d),\mathrm{LS}}^2$ being low as a reassurance that for a significant portion of states the protocol is sample efficient, and hence as supporting evidence for our initial claim. More specifically, if the locally scrambled shadow norm of an observable is smaller than some bound $B$, by Markov's inequality we can conclude that for a fraction $1-\delta$ of all states in any $1$ design, its shadow norm is smaller than $B/\delta$. In Appendix \ref{app:av_shadow_proof} we employ recently developed methods for the anti-concentration of random quantum circuits to bound the locally scrambled shadow norm~\cite{hunter2019unitary,dalzell2022random,barak2020spoofing}.
In particular, the locally scrambled shadow norm reduces to so-called second moments over the ensemble of random $2$-local circuits.
Using Weingarten calculus, we map these second moments to a partition function of an effective statistical mechanics model.
We find the following theorem:

\begin{thm}\label{thm:av_shadow_norm_UB}
  If $d= \Theta(\log(n))$, we have for any traceless observable $O$
    \begin{equation}
        ||O||_{s(d),\mathrm{LS}}^2\leq 2||O||_{F}^2\left(1+\frac{1}{n^{\mathcal O(1)}}\right) .
    \end{equation}
    Furthermore, let $l$ be the maximum distance between two sites on which $O$ is supported, if $l\leq \mathcal O(\log(n))$,
    \begin{equation}
      ||O||_{s(d),\mathrm{LS}}^2\leq n^{\mathcal O(1)}2^{-n}||O||_F^2.
    \end{equation}
\end{thm}
The constants and precise expressions are given in the proofs. The first bound in this theorem shows that for a large fraction of input states, a logarithmic depth shallow shadows scheme is sufficient to achieve the same sample efficiency as the global Clifford scheme of HKP (as in this case for traceless $O$, $||O||_F^2\leq ||O||_{s(\infty)}^2\leq 3||O||_F^2$, see Ref.~\cite{Huang2020}), up to a small correction which vanishes as the system size grows. The second bound shows that additionally, some of the desirable properties of the $d=0$ case are preserved up to and including logarithmic depth, namely, the sample complexity for local observables does not, for typical input states, grow exponentially in the system size. 
\section{Numerics}\label{sec:numerics}
 \emph{Numerics.} We now present numerical output generated by %(inefficiently) 
classically simulating the data acquisition part of our scheme (the quantum processing component that is to be implemented on a quantum computer) followed by the implementation of our protocol to estimate expectation values. The code used to produce this data is available at \cite{code}.
%a few practical examples of estimations, along with computations of shadow norms. 
In the following, we choose $\rho$ to be the GHZ state on $n=8$ qubits, unless otherwise specified. We start by estimating the fidelity of the state with itself, i.e., pick $O=|\mathrm{GHZ}\rangle\langle \mathrm{GHZ}|$. Since $O$ has full support and $\pnorm{O}{\mathrm{F}}=1$, the $d=0$ scheme is expected to perform poorly and the $d=\infty$ scheme is expected to perform well. 
For each depth, we obtain 100 independent estimates with $N=1000$ samples each.

\cref{fig:fidelity_estimation} shows the performance of the protocol for various depths.
In particular, we see that the spread in estimated values depends on the depth and is consistent with both the empirically calculated variance and the theoretically computed shadow norms. It is apparent that the performance becomes dramatically better from $d=0$ to $d=2$ (decreasing shadow norm and variability in estimates). Thereafter, the improvement in performance for every additional layer is much smaller, indicating that the global Cliffords scheme's performance is already matched after a short depth. 
\begin{figure}[h!]
    %\centering
         \includegraphics[width=.4\textwidth]{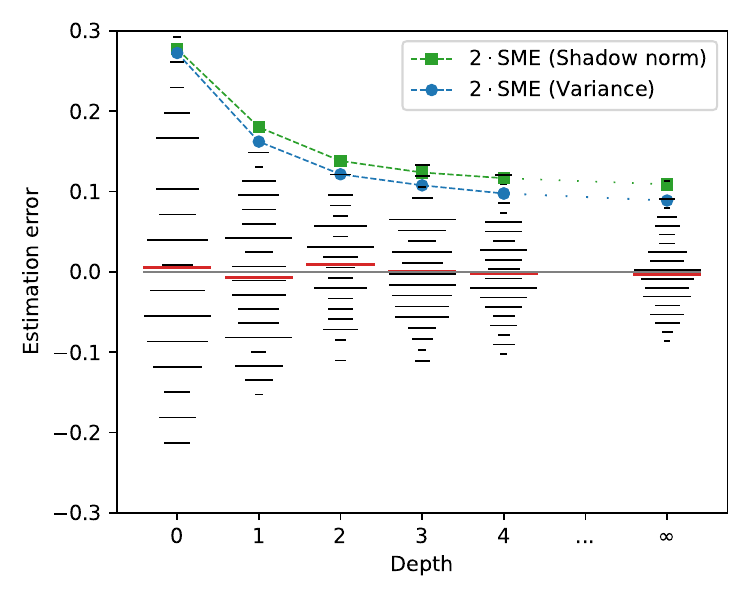}
    \caption{Performance of the protocol for $O=\rho=|\mathrm{GHZ}\rangle\langle\mathrm{GHZ}|$ for various depths. For each depth we show the spread of errors in estimates based on 1000 samples as well as a bound for twice the \emph{standard mean error} (SME) corresponding to $\sim$ 95\% confidence interval assuming Gaussian distributed errors.
    Each black line has a length proportional to the fraction of estimates with similar estimation error.
    The SMEs are computed using the empirical variance and the shadow norm.
    The red lines corresponds to the estimates obtained using the whole set of 100,000 samples.}
    \label{fig:fidelity_estimation}
\end{figure}
We demonstrate the performance of the protocol at various depths in estimating single Pauli operators. 
\cref{fig:pauli_shadow_norm} shows the squared shadow norm of the Pauli $Z^{\otimes k }$ on $n=20$ qubits for various depths. As proven in Section V of the Supplemental Material, the squared Pauli shadow norm equals the inverse eigenvalue $1/t_{\lambda,d}$ and is therefore efficiently computable. 
Since all the Pauli operators considered have support on more than one qubit, intermediate depths perform better than $d=0$, and except for the fully supported case, better than $d=\infty$, but even in this case, the shadow norm quickly approaches the limiting value.
 \begin{figure}
    \centering
    \includegraphics[width=.4\textwidth]{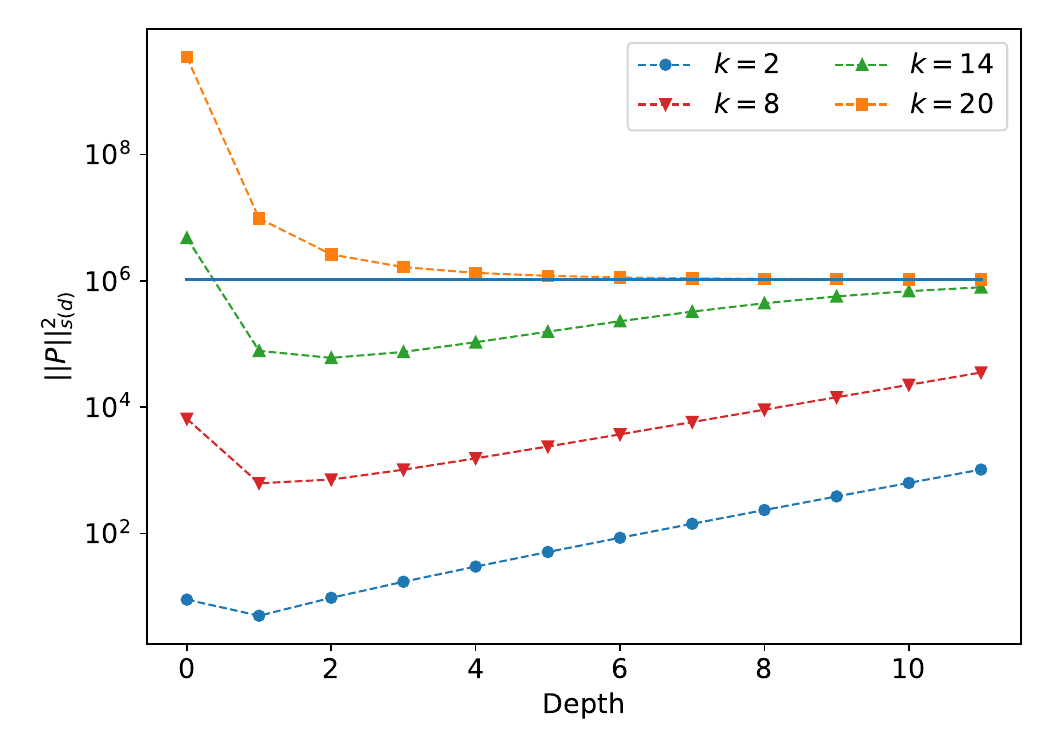}
    \caption{Squared shadow norm of $Z^{\otimes k}$ for various depths and $n=20$ qubits. The straight line corresponds to the 
    limit
    %limiting value 
    $2^n+1$ for $d\to\infty$.}
    \label{fig:pauli_shadow_norm}
\end{figure}

\begin{figure}[ht!]
    \centering
         \includegraphics[width=.4\textwidth]{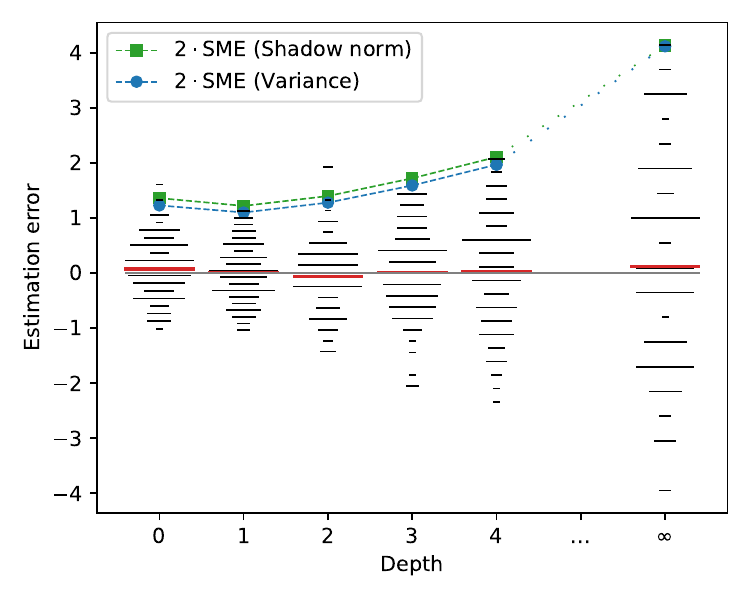}
    \caption{Performance of the protocol for $O=H$ and $\rho=|\mathrm{GHZ}\rangle\langle\mathrm{GHZ}|$ for various depths. For details refer to \cref{fig:fidelity_estimation}.}
    \label{fig:hamiltonian_estimation}
\end{figure}
Next, we look at the estimation of a sparse observable, namely the Hamiltonian 
$H=\sum_i (Z_{i-1}Z_iZ_{i+1}+X_i)$, with periodic boundary conditions. We have picked a 3-local Hamiltonian to exemplify the effect of depth at intermediate Pauli supports, analogous to \cref{fig:pauli_shadow_norm}. 
In \cref{fig:hamiltonian_estimation}, we demonstrate the performance of the protocol in estimating $\langle H\rangle$. 
As expected, low depth is better suited for this observable, $d=1$ performs slightly better than $d=0$ because of the 3-local term, and it is evident that the global $d=\infty$ protocol is ill-suited for the estimation of this observable. Overall, we observe a behavior analogous to the one of single qubit Paulis. 

\section{Discussion}
In this work we provide a protocol for estimating properties of 
complex quantum systems from randomized measurements.
In particular, we consider measurements in a basis generated by random quantum circuits of $2$-local Clifford gates of arbitrary depth $d$. This ensemble sits between the two extremes proposed in the seminal work Ref.~\cite{Huang2020}: uniformly random single qubit Clifford unitaries and uniform global Clifford unitaries. We show that the eigenvalues of the corresponding measurement channel have an intuitive interpretation as probabilities and that it is possible to express them as MPS, and provide a variational procedure to obtain a MPS approximating the eigenvalues of the inverse measurement channel. This procedure is heralded in the sense that its accuracy can be verified and the error of the approximation can be efficiently measured. Overall, this allows us to compute expectation values of a large class of experimentally relevant observables, namely, those that can be written as a polynomial bond dimension MPO, as well as to estimate the number of samples necessary 
to achieve a certain accuracy in the estimation. Up to the inversion of the measurement channel, which has to be performed only once for any given depth and number of qubits, our protocol is computationally efficient in the regime $d=\mathcal O(\log(n))$. 
We gather evidence that in this regime our protocol yields essentially the same guarantees on the sampling complexity as the global Cliffords scheme.
In particular, we find that for fully supported observables the locally scrambled shadow norm is essentially the same at a depth $d\sim\log(n)$ as it is at infinite depth, and at the same time it remains controlled for locally supported observables, which is not true for the global Clifford scheme.
This is shown with the same techniques that are used to prove that log-depth random quantum circuits anti-concentrate~\cite{dalzell2022random,barak2020spoofing,SupremacyReview}, a 
key property for quantum advantages, and we expect that a similar behaviour is true for the worst-case shadow norm, as this is simply a third moment of the circuit. A rigorous proof of this is much more involved and we leave this to future work.
Moreover, it has been shown in Ref.~\cite{deshpande2021tight} that the ensemble ``severely'' fails to anti-concentrate at 
sub-logarithmic depth, we then expect the ensemble at lower depth to have properties closer to the single qubits Cliffords and admit large sample complexity for non-local observables. This behaviour suggests a phase transition from local-to-global happening at log-depth, we then expect this regime to be a ``sweet spot" that inherits the benefits of the global Cliffords scheme while requiring much fewer resources, being efficiently implementable, and retaining the desirable locality properties of the local Cliffords schemes.

Many questions and problems remain open: Low bond dimension MPS representations of the inverse measurement channel appear to exist based on numerical evidence, nevertheless, there is to date
no theoretical guarantee of their existence, nor of their expected bond dimension. It is expected that a proof of this fact would shed light on a more guided manner of finding the inverse. Furthermore, while our bounds on the locally scrambled shadow norm provide compelling evidence that this protocol is useful at log-depth, a rigorous proof of sample efficiency for large classes of observables is missing. Finally, we provide ways of estimating the sample complexity for individual observables, we are confident that the computations involved in these procedures can be made more efficient by MPS bond dimension reduction techniques.

% -------------------------------------------------------------
% ACKNOWLEDGEMENTS
% -------------------------------------------------------------

\section{Acknowledgements}
We thank Ingo Roth and David Wierichs for fruitful discussions. We also acknowledge a discussion with two groups of researchers: On the one hand,
 Ahmed A. Akhtar, Hong-Ye Hu, Yi-Zhuang You, and on the other hand  Mirko Arienzo, Markus Heinrich, and Martin Kliesch that both independently and concurrently investigated intermediate classical shadows schemes. We thank Stefano Mangini for making us aware of a minor mistake.
HP acknowledges the Centre for Quantum Software and Information at the University of Technology Sydney and Michael Bremner for hosting him as a visiting scholar.
We thank the BMBF (Hybrid, MuniQC-Atoms, DAQC)
the MATH+ Cluster 
of Excellence, the Einstein 
Foundation (Einstein Unit on Quantum Devices), 
the QuantERA (HQCC), the Munich Quantum Valley (K8),
and the 
DFG (CRC 183, EI 519 20-1, EI 519/21-1) 
for support.

\appendix
\onecolumngrid
% -------------------------------------------------------------
% APPENDIX A
% -------------------------------------------------------------

\section{Proof of \cref{lem:measurement_map_action}}\label{app:proof_lem1}
Let us begin by noting that for all $d$ (including $d=0$ and $d=\infty$) the ensemble $\sensemble{d}$ is Pauli invariant in the sense that for any Pauli $S\in \pauli{n}$, and any Clifford $U\in \seset{d}$, we have $\seprob{d}(PU)=\seprob{d}(UP)=\seprob{d}(U)$. Let us first prove the following useful result.

\begin{lem}\label{lem:av_prod_paulis}
Let $\mathcal U$ be a Pauli invariant ensemble and let $P_1,\dots ,P_k\in \mathcal P_n$ be a collection of $k$ Pauli operators such that $\prod_{i=1}^k P_i\neq \pm \mathbb{I}$. Then
\begin{equation}
\expval{U\sim\mathcal U}{\prod_{i=1}^k \langle 0| UP_iU^{\dagger}|0\rangle}= 0.
\end{equation}
In addition, if $\prod_{i=1}^k P_i=\mathbb{I}$, then
\begin{equation}
    \expval{U\sim\mathcal U}{\prod_{i=1}^k \langle 0| UP_iU^{\dagger}|0\rangle} =\prover{U\sim \ensemble}{\forall\,i,\,UP_iU^{\dagger}\in \pm\mathcal Z}.
\end{equation}
\end{lem}
\begin{proof}
We begin by noticing that for any Pauli $P$, $\langle0|P|0\rangle=0$ unless $P\in \pm \mathcal Z$, and in this case its value is $\pm 1$. This means that unless $P_1,\dots,P_k$ all commute,  $\prod_{i=1}^k \langle 0| UP_iU^{\dagger}|0\rangle=0$ for any Clifford $U$, since for this product not to vanish, we must have $UP_iU^{\dagger}\in \pm Z$ for all $i$. Since a product of Hermitian commuting matrices is Hermitian, $\prod_{i=1}^k P_i\neq \pm i\, \mathbb{I}$.
 Suppose $\prod_{i=0}^k P_i \neq \pm \mathbb{I}$, and thus $P_1\neq \pm P_2\dots P_k$, hence, there exists $Q$ such that $[P_1,Q]=0$ and $\{Q,P_2\dots P_k\}=0$. For $Q$ to anti-commute with the products of $P_2$ to $P_n$, it must anti-commute with an odd number of them. We then have 
\begin{equation}
    \expval{U\sim\mathcal U}{\prod_{i=1}^k \langle 0| UP_iU^{\dagger}|0\rangle}=\expval{U\sim\mathcal U}{\prod_{i=1}^k \langle 0| UQP_iQU^{\dagger}|0\rangle}=-\expval{U\sim\mathcal U}{\prod_{i=1}^k \langle 0| UP_iU^{\dagger}|0\rangle},
\end{equation} 
where we have used Pauli invariance in the first equality.
Suppose now $\prod_{i=0}^k P_i=\mathbb{I}$. Since $\prod_{i=1}^k UP_iU^{\dagger}=\mathbb{I}\in +\mathcal Z$, if $UP_i U^{\dagger}\in\pm \mathcal Z$ for all $i$, only an even number of $UP_iU^{\dagger}$ can be in $-\mathcal Z$, then
\begin{equation}
    \prod_{i=1}^k \langle 0|UP_iU^{\dagger}|0\rangle=\begin{cases}
    1 &\textrm{ if } UP_iU^{\dagger}\in \pm \mathcal Z \,\forall \, i=1,\dots, k\\0&\textrm{ otherwise. }
    \end{cases}
\end{equation}
By using the definition of expectation value, we get
\begin{equation}
    \expval{U\sim\mathcal U}{\prod_{i=1}^k \langle 0| UP_iU^{\dagger}|0\rangle} =\prover{U\sim \ensemble}{\forall\,i,\,UP_iU^{\dagger}\in \pm\mathcal Z}.
\end{equation}
\end{proof}
We now move on to proving \cref{lem:measurement_map_action}.
\begin{proof} From \cref{eq:measurement_map}, 
we have 
\begin{equation}
    \smeasmap{\rho}{d}=\expval{U\sim \sensemble{d}}{\sum_{b \in \bitstring{n}} \matrixel{b}{U \rho U^{\dagger}}{b} U^{\dagger}\ketbra{b}{b}U}=   2^n\expval{U\sim \sensemble{d}}{ \matrixel{0}{U \rho U^{\dagger}}{0} U^{\dagger}\ketbra{0}{0}U},
\end{equation}
where we have used that the ensemble $\sensemble{d}$ is Pauli invariant and hence every term in the sum must have the same value, since $|0\rangle$ differs from $|b\rangle$ only by a Pauli string. For $P^{\lambda},P^{\lambda'} \in \pauli{n}$, expanding $\smeasmap{P}{d}$ in the Pauli basis, we can write the coefficient corresponding to $P^{\lambda'}$ as 
\begin{equation}
    \frac{1}{2^{n}}\tr{P^{\lambda'}\smeasmap{P^{\lambda}}{d}}=  \expval{U\sim \sensemble{d}}{ \matrixel{0}{U P^{\lambda} U^{\dagger}}{0} \matrixel{0}{U P^{\lambda'} U^{\dagger}}{0} }.
\end{equation}
By \cref{lem:av_prod_paulis}, this vanishes unless $P^{\lambda'}= \pm P^{\lambda}$, but since $-P^{\lambda}\notin P_n$, it vanishes unless $P^{\lambda'}=P^{\lambda}$, then
\begin{equation}
    \measmap{P^{\lambda}}=t_{\lambda,d} P^{\lambda}
\end{equation}
with, again by \cref{lem:av_prod_paulis},
\begin{equation}\label{eq:t_exp_val}
    t_{\lambda,d}=\frac{1}{2^n}\tr{P^{\lambda}\measmap{P^{\lambda}}}=\expval{U\sim \ensemble}{ \matrixel{0}{U P^{\lambda} U^{\dagger}}{0}^2}= \prover{U\sim \sensemble{d}}{UP^{\lambda}U^{\dagger}\in \pm\mathcal Z}.
\end{equation}
\end{proof}

% -------------------------------------------------------------
% -------------------------------------------------------------
% -------------------------------------------------------------
% \section{Other interpretations of the transition probability}\label{app:transition prob}

% We note that in addition to its interpretation as a probability as per \eqref{eq:t as a prob}, $\tprob{P}$ can be interpreted as an expectation value as
% \begin{align}
%     \tprob{P}=\expval{U\sim \sensemble{d}}{\abs{\bra{0} UPU^{\dagger} \ket{0}}},
% \end{align}
% and also as the variance of the random variable $\bra{0} UPU^{\dagger} \ket{0}$.

% -------------------------------------------------------------
% -------------------------------------------------------------
% -------------------------------------------------------------
\section{Proof of \cref{lem:meas_map_comp}}\label{app:proof_lem2}

\vspace{0.5cm}
\begin{figure}[h]
    \centering
    \includegraphics[scale=.3, trim= 0 100 400 100]{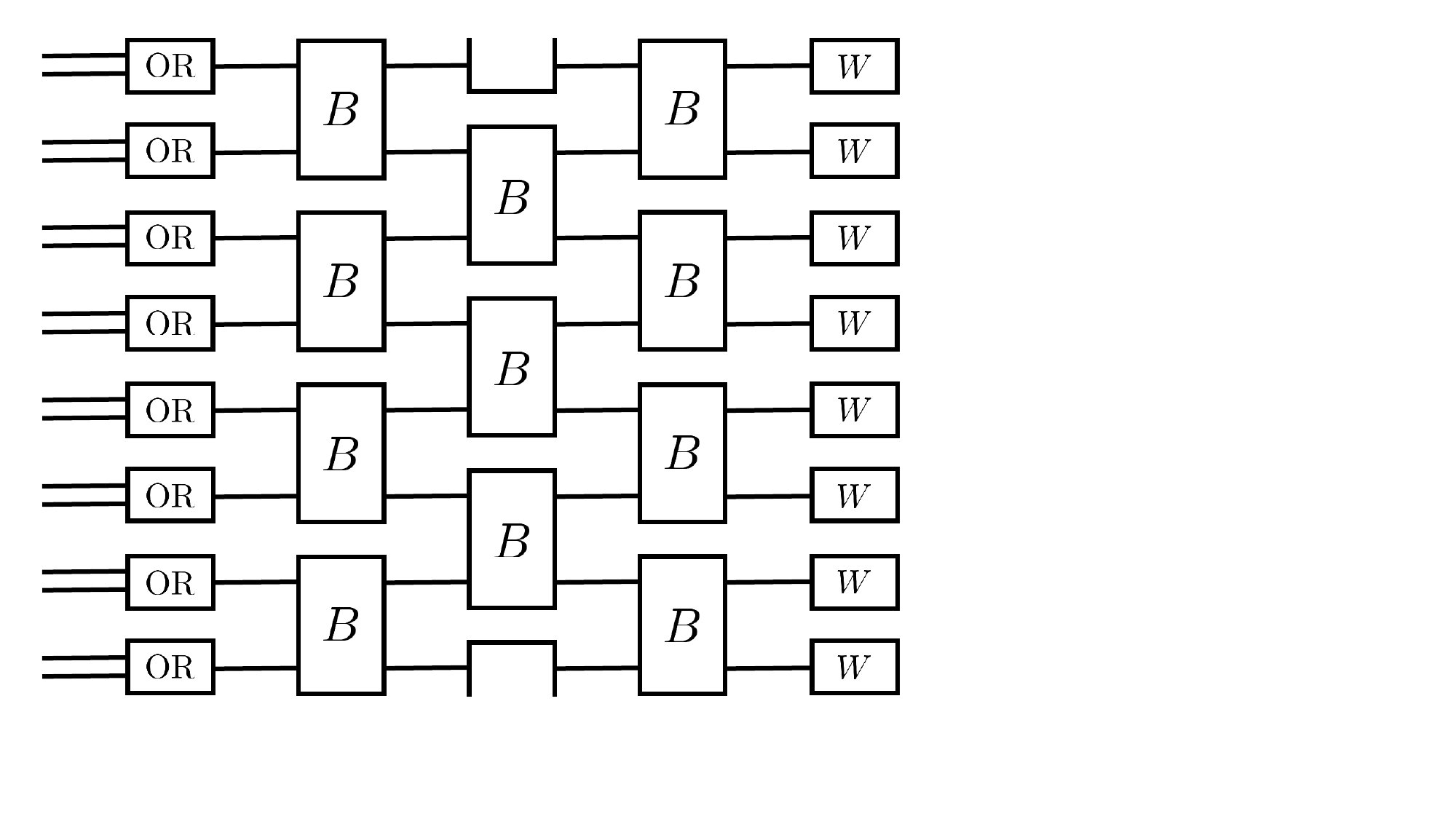}
    %trim= 1 2 3 4 ‘crops’ the picture by 1bp at the left, 2bp at the bottom, 3bp on the right and 4bp at the top
    \caption{Tensor network for computing $\tprob{P^{\lambda}}$ for $n=8$ qubits and depth $d=3$. This can be constructed from the random Clifford architecture (c.f. \cref{fig:brickwork}) by replacing single qubit random Clifford gates by the $\textsc{or}$ gate, replacing $2-$qubit random Clifford gates by the $B$ matrix and ending with a $W$ vector which evaluates to $1/3^a$ for index value $a\in \set{0,1}$.}
    \label{fig:tprob_TN}
\end{figure}

\begin{proof}
We define the map $g: \pauli{n} \rightarrow \bitstring{n}$ such that the $k^{\rm th}$ entry of $g(P)$ is $0$ if and only if the $k^{\rm th}$ tensor factor of $P$ is the identity. We call $g(P)$ the signature of $P$. We note that $g(P^{(x,z)})=(\logicnor{x_1,z_1},\ldots,\logicnor{x_n,z_n})$. From \cref{eq:t as a prob} and by noting that the $0^{\rm th}$ layer of gates locally scrambles/symmetrizes the Paulis $X,Y$ and $Z$, it is clear that $\tprob{P}$ can only depend on the signature of $P$. To construct a tensor network that evaluates $\stprob{P^{\lambda}}{d}$ (for an input leg assignment of $\lambda$), we first act with $\textsc{or}$ gates to map $\lambda$ to the signature of $P^{\lambda}$ (c.f. \cref{fig:tprob_TN}). Then we act with a $4 \times 4$ real matrix $B$ in a brickwork pattern corresponding to the locations of the $2$-qubit random Clifford gates in the brickwork architecture (c.f. \cref{fig:brickwork}). The columns of the matrix $B$,
\begin{align*}
   B:= \begin{bmatrix}
1 & 0 & 0 & 0 \\
0 & 0.2 & 0.2 & 0.2 \\
0  & 0.2 & 0.2 & 0.2  \\
0 & 0.6 & 0.6 & 0.6 
\end{bmatrix},
\end{align*}
can be interpreted as conditional probability distributions. Indexing the rows and columns of $B$ by the bit pairs $g,g' \in \bitstring{2}$, respectively, we interpret $B_{g,g'}$ as the probability that a uniformly sampled $2-$qubit Clifford gate $U$ conjugates a Pauli with signature $g'$ to a Pauli with signature $g$. Hence, after applying $d$ layers of $B$ matrices, if we label the input legs of the tensor network by $\lambda\in \bitstring{2n}$ and the output legs by $\gamma\in \bitstring{n}$, the tensor element will correspond to the probability $$\prover{U\sim \sensemble{d}}{g(UP^{\lambda}U^{\dagger})=\gamma}.$$ 
We need to multiply this by the conditional probability $$\prover{U\sim \sensemble{d}}{UP^{\lambda}U^{\dagger}\in \pm \set{\mathbb{I},Z}^{\otimes n}~|~g(UP^{\lambda}U^{\dagger})=\gamma}$$ and sum over all $\gamma$ to arrive at the probability in \cref{eq:t as a prob}. Checking the condition $UP^{\lambda}U^{\dagger}\in \pm \set{\mathbb{I},Z}^{\otimes n}$ can be done by checking if each local Pauli factor is either $\mathbb{I}$ (always the case when its associated signature is $0$) or $Z$. Given that an output Pauli factors has signature $1$, it is a Pauli $Z$ with probability $1/3$ over the randomness of $U$. This can easily be checked by noting that this property is inherited from the last local (one or two qubit) component of $U$ to act on the Pauli factor. Thus, $t_{\lambda,d}$ (c.f. \cref{eq:t as a prob}) can be computed by contracting each output leg by the vector $W=[1 ~ 1/3]$ which applies a factor of $1$ if the local signature is $0$ and a factor of $1/3$ if the signature is $1$. The resulting tensor network evaluates $t_{\lambda,d}$, however, we make some further simplifications. First, we note that since the last three columns of $B$ are identical, $B_{g,g'}$ only depends on $\logicor{g'}$ in the sense that $B_{g,g'}=B'_{g,\logicor{g'}}$ where
\begin{align*}
   B':= \begin{bmatrix}
1 & 0\\
0 & 0.2\\
0  & 0.2\\
0 & 0.6
\end{bmatrix}.
\end{align*}

This simplifies \cref{fig:tprob_TN} to the tensor network given in \cref{fig:tensor_definitions}. 
\begin{figure}
     \centering
     %\begin{subfigure}[b]{0.45\textwidth}
     \subfloat[]{
         \centering
         \includegraphics[scale=0.3, trim= 20 50 200 0,clip]{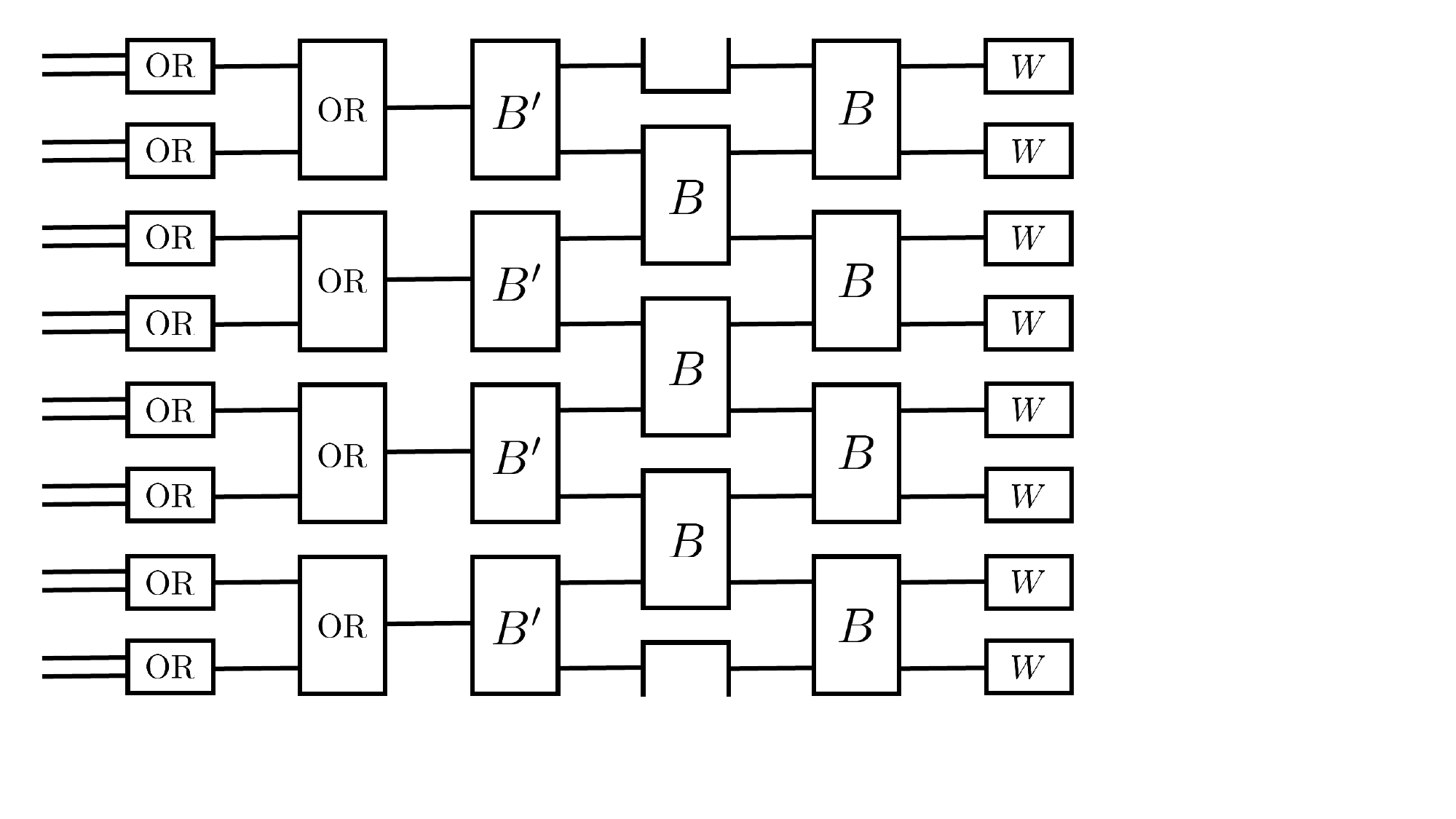}
         %\caption{}
         }
    %trim= 1 2 3 4 ‘crops’ the picture by 1bp at the left, 2bp at the bottom, 3bp on the right and 4bp at the top
         %\caption{$y=x$}
         %\label{fig:tprob_TN_simpler}
     %\end{subfigure}
     %\hfill
     %\begin{subfigure}[b]{0.45\textwidth}
     \subfloat[]{
         \centering
         \includegraphics[scale=.3, trim= 20 50 200 0,clip]{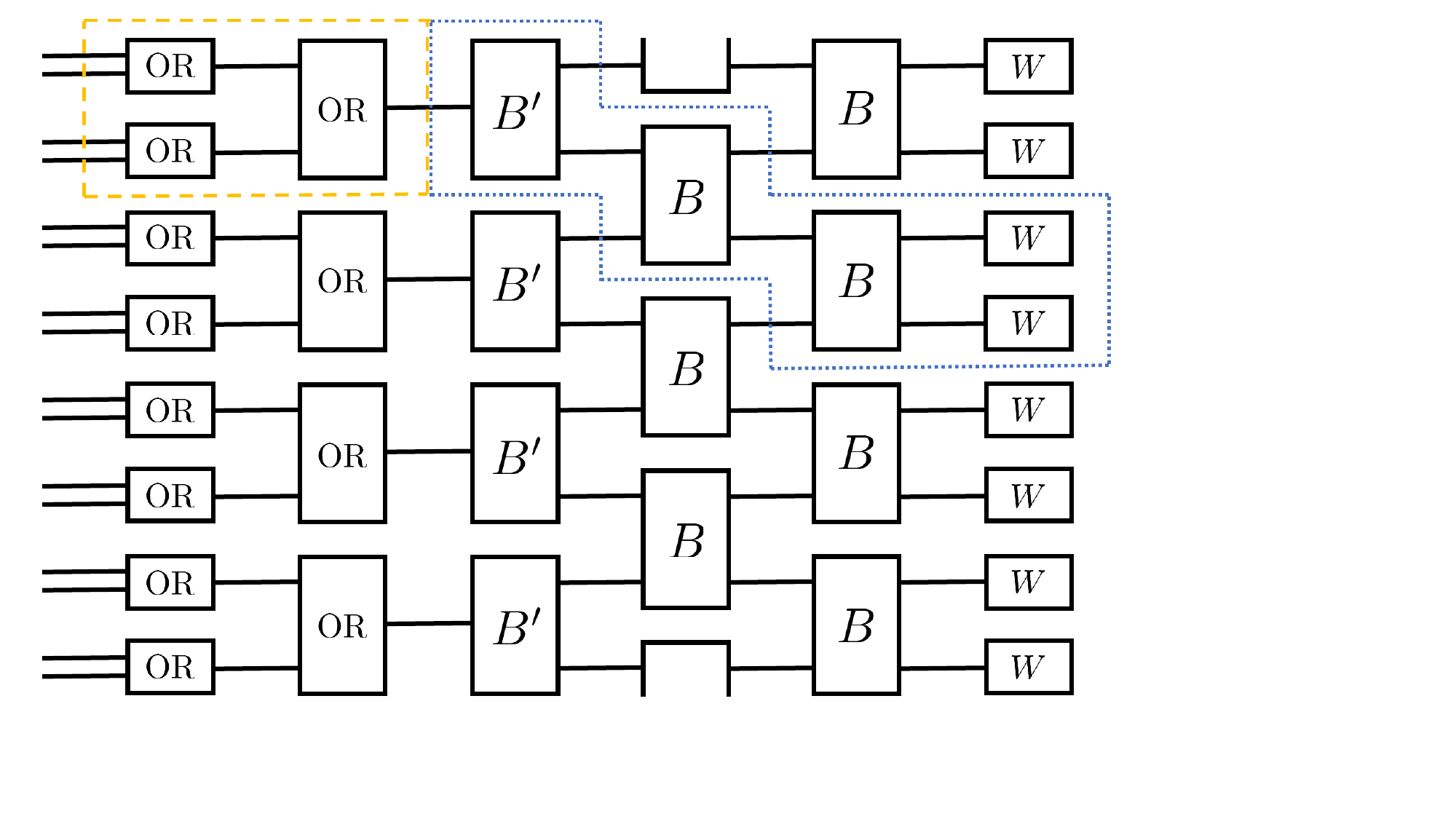}
    %trim= 1 2 3 4 ‘crops’ the picture by 1bp at the left, 2bp at the bottom, 3bp on the right and 4bp at the top
         %\caption{$y=3sinx$}
         %\label{fig:tprob_TN_simpler_grouped}
        %\caption{}
        }
     %\end{subfigure}
        \caption{(a) A simpler tensor network for computing $t_{\lambda,d}$ for $n=8$ qubits and depth $d=3$. (b) Shows groupings of components that will be used for the form presented in \cref{fig:tprob_TN_MPS}.}
        \label{fig:tensor_definitions}
\end{figure}
Using the definitions presented in \cref{fig:tensor_component_definitions}, we can draw this tensor network as per \cref{fig:tprob_TN_MPS}.

\begin{figure}
     \centering
     \subfloat[]{
         \centering
         \includegraphics[scale=0.3, trim= 20 300 200 0,clip]{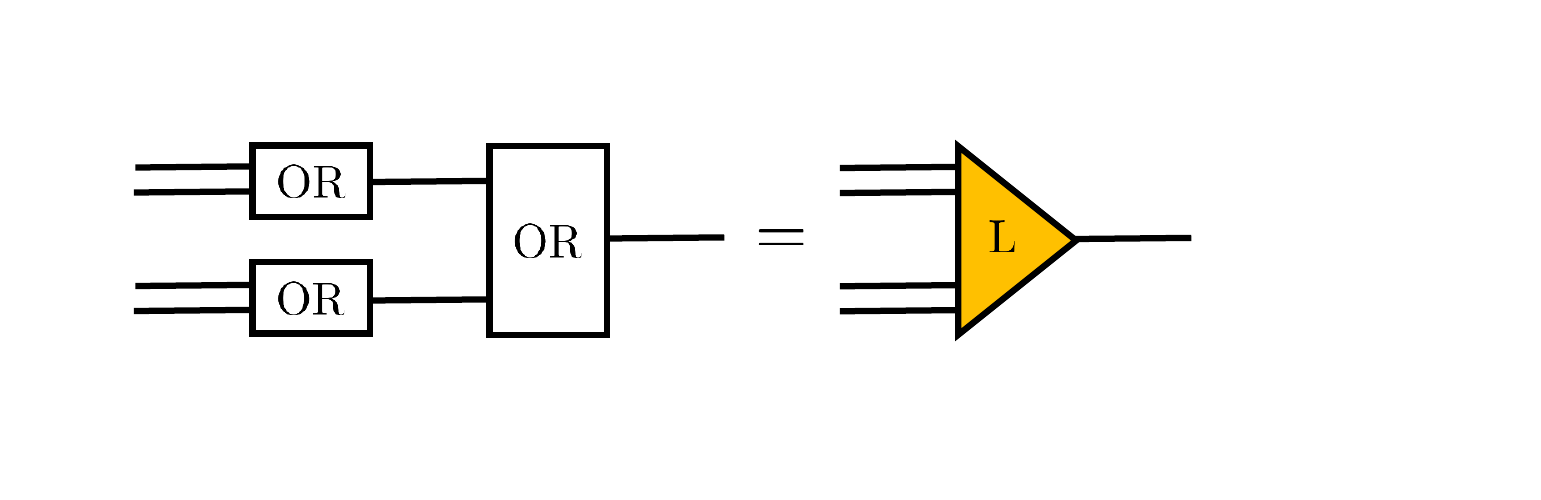}
    %trim= 1 2 3 4 ‘crops’ the picture by 1bp at the left, 2bp at the bottom, 3bp on the right and 4bp at the top
         %\caption{$y=x$}
         %\label{fig:tprob_TN_simpler}
   }
     %\hfill
     \subfloat[]{
         \centering
         \includegraphics[scale=.3, trim= 20 150 200 0,clip]{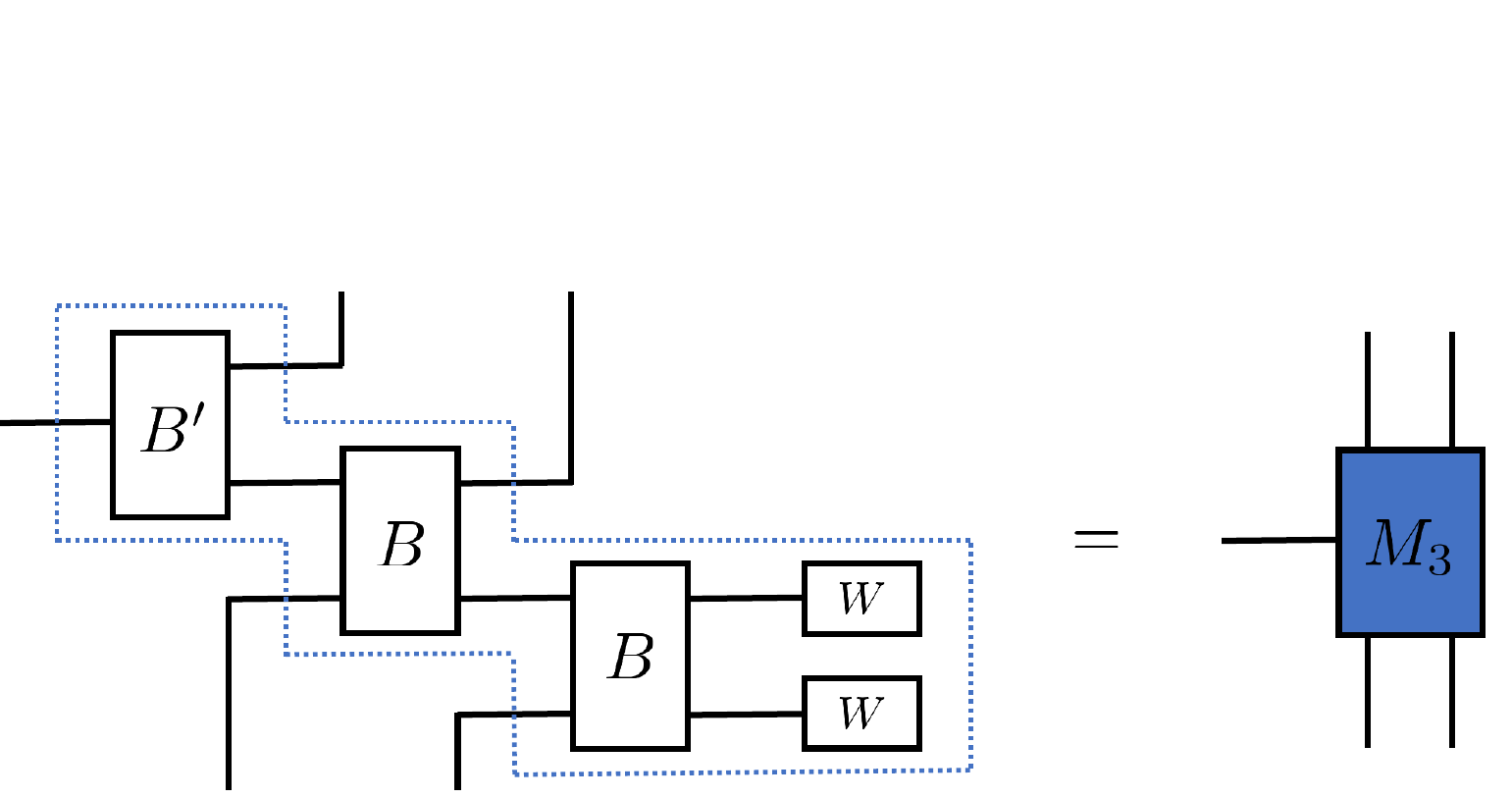}
    %trim= 1 2 3 4 ‘crops’ the picture by 1bp at the left, 2bp at the bottom, 3bp on the right and 4bp at the top
         %\caption{$y=3sinx$}
         %\label{fig:tprob_TN_simpler_grouped}
     }
        \caption{Representation of grouped components used in \cref{fig:tprob_TN_MPS}.}
        \label{fig:tensor_component_definitions}
\end{figure}

\end{proof}

\section{MPS Inversion algorithm}\label{app:variational_inverse}
In this section, we describe our algorithm to produce an approximate MPS representation for the inverse channel $\mathcal M_d^{-1}$ from the MPS representation of $\mathcal M_d$. Before starting, we remark that all the code used in conjunction with this project is publicly available at \cite{code}. It includes the implemention of the MPS inversion algorithm discussed in this section.

As we discussed, all unique eigenvalues $t_{\lambda,d}$ of the measurement channel on $n$ qubits can be encoded in an MPS with $n/2$ physical legs, bond dimension $2^{d-1}$ and physical dimension $2$, where the $n/2$ tensors are given by the right hand side of \cref{fig:tensor_component_definitions}. The goal here is to find an algorithm outputting an analogous representation for the eigenvalues $1/t_{\lambda,d}$ of the inverse measurement map $\mathcal M_d^{-1}$.
The algorithm should solve the following abstract task efficiently: Given an MPS representation for a vector $m\in {\mathbb R^{D}}^{\otimes N}$, find an MPS representation for the element-wise inverse vector which we denote by $v$. That is, $v(x)=1/m(x)$, where $x\in \{0,\dots, D-1\}^{N}$ and $m(x)$ is the $x$-th component of the vector $m$. In the following we are going to set $D=2$ and $N=n/2$, corresponding to our problem of representing $1/t_{\lambda,d}$. 
We say that $v$ is an $\epsilon$-approximate (element-wise) inverse of $m$ if 
\begin{equation}
    \sqrt{\sum_{x\in\{0,1\}^N}|m(x)v(x)-1|^2}\leq \epsilon.
\end{equation}
In practice, we aim for a good approximate inverse rather than an exact inverse of $m$.

What is particular to our problem is that $m(x)$ is actually translationally invariant: Let $m(x)$ be the vector corresponding to the unique eigenvalues of $\mathcal M$, then consider $m(T(x))$, where $T$ translates the bit string by one site. On the level of $t_{\lambda,d}$, $T$ corresponds to translating the corresponding Pauli operator $P^{\lambda}$ by two sites, which does not change the eigenvalues $t_{\lambda,d}$ and so $m(T(x))=m(x)$. Note that this also holds for the eigenvalues $1/t_{\lambda,d}$ of $\mathcal M^{-1}$ and hence for the vector $v(x)=1/m(x)$. Since a translationally invariant vector admits a translationally invariant MPS (TI-MPS) representation (c.f. e.g. \cite{phdthesis}, Theorem 3), we know a TI-MPS for $1/t_{\lambda,d}$ must exist.
Nevertheless, we remark that forcing the solution to be a TI-MPS might lead to a higher bond dimension than outputting a non TI representation. Moreover, an optimization not imposing translational invariance appears to work better in practice as the search space for an accurate solution is larger and the optimization landscape more favorable. For these reasons we will for the moment abandon translational invariance, nevertheless at the end of this section we will argue that guiding the algorithm to a translationally invariant inverse can be beneficial.

We now present our MPS inversion algorithm.
Let $v$ have an MPS representation $V=\{V_0^j,V_1^j\}_{j=1}^N$, where the $V_k^j$ are $\chi_V\times \chi_V$ matrices, and similarly let $m$ have an MPS representation $M=\{M_0^j,M_1^j\}$ where the $M_k^j$ are $\chi_M\times \chi_M$ matrices. From now on we will assume $m$ and $v$ to be real vectors, and the tensors in the MPSs to be defined over the reals. All the results here can be easily generalized to the complex case. We can find an approximate element-wise inverse vector $v$ by minimizing the cost function
\begin{equation}\label{eq:cost_function}
    C_M(V)=\sum_{x\in\{0,1\}^N}(m(x)v(x)-1)^2= \tr{\prod_{j=1}^{N}\sum_{k=0}^1{M_k^j}^{\otimes 2}\otimes {V_k^j}^{\otimes 2} } -2 \tr{\prod_{j=1}^N\sum_{k=0}^1{M_k^j}\otimes {V_k^j}}+2^N
\end{equation}
Notice that computing the cost function only involves multiplying $n$ matrices of size $\chi_V^2\chi_M^2$ or less. Since $M$ is fixed, we will drop the subscript $M$ from $C_M(V)$ from now on.

 We minimize this cost function using a form of \emph{cyclic coordinate descent}, that is, starting from a random ansatz, we optimize the cost function by varying one of the matrices in $\{V\}$, while keeping the others fixed, and repeat the process with all the matrices in $\{V\}$. Sweeping over all the matrices in this way the algorithm converges to a minimum. For convenience, we report \cref{alg:MPSinv} here.
%XXXXXXXXXXXXXXXXXXXXXXXXXXXXXXXXXXXXXXXXXXXXXXXXXXXXXXXXXXXXX
%Inverse MPS Alg
%XXXXXXXXXXXXXXXXXXXXXXXXXXXXXXXXXXXXXXXXXXXXXXXXXXXXXXXXXXXXX
\begin{algorithm}[H]
\caption*{Algorithm for computing an approximate inverse MPS}
\begin{algorithmic}[1]
\algrenewcommand\algorithmicrequire{\textbf{Input:}}
\algrenewcommand\algorithmicensure{\textbf{Output:}}
\Require MPS $\{M_0^j,M_1^j\}_{j=1}^{N}$, ansatz $\{V^j_0,V^j_1\}_{j=1}^{N}$, threshold $\epsilon$.
\Ensure Approximate inverse $\{V^j_0,V^j_1\}_{j=1}^{N}$
\Statex
\While{$C(\{V\})>\epsilon$}
\For{$j=1,\dots N$}
\For{$k=0,1$}
\State $V^j_k \xleftarrow{} \underset{V_{k}^j}{\mathrm{argmin}} \, C(\{V\})$
\EndFor
\EndFor
\EndWhile
\State{\Return $\{V^j_0,V^j_1\}_{j=1}^{N}$}
\end{algorithmic}
\end{algorithm}
%XXXXXXXXXXXXXXXXXXXXXXXXXXXXXXXXXXXXXXXXXXXXXXXXXXXXXXXXXXXXXXXXX

Our choice of optimization algorithm is motivated by the fact that in our case the local optimization, step 4 in \cref{alg:MPSinv}, is particularly simple as demonstrated below. This is because the cost function is a bounded quadratic form given that all matrices except one are fixed. For the optimization, this implies that the minimum at each optimization step is unique, and hence the procedure is guaranteed to converge to a stationary point of the cost function \cite{GRIPPO2000127}. 

Define the function
\begin{equation}
    C^{j}(X_0, X_1)=C(V_0^1,V_1^1,V_0^2,\dots, V_1^{j-1},X_0, X_1, V_0^{j+1}, \dots, V_1^N).
\end{equation}
as the function that we get when we fix all matrices except $V_{0}^j, V_{1}^j$. This function can be written as a quadratic form.
\begin{lem}
\begin{equation}
    C^j(X_0, X_1)=C_0^j(X_0)+C_1^j(X_1)
\end{equation}
\begin{equation}
    C_{k}^j(X)=\langle X|A_k^{j}|X\rangle+\langle B_k^{j}|X\rangle +2^{N-1}
\end{equation}
where $|X\rangle=\sum_{r,s=1}^{\chi_V}X_{r,s}\ket r \ket s$ is the vectorization of $X$, $A_k^{j}$ is a positive semi-definite $\chi_V^2\times \chi_V^2$ matrix, and $|B_k^{j}\rangle$ is a vector of size $\chi_V^2$. Furthermore, $A_k^{j}$ and $|B_k^{j}\rangle$ are efficiently computable. 
\end{lem}
\begin{proof}
Without loss of generality due to periodic boundary conditions, consider $j=1$. For $x\in\{0,1\}^n$, define $V_x^{(a,b)}=\prod_{j=a}^{b} V_{x_j}^j$ and similarly for $M_x^{(a,b)}$, then we have by expanding the product of sums
\begin{equation}
\begin{aligned}
    C_1(V_0^1,V_1^1)&=\sum_{x\in\{0,1\}^{N}}\tr{M_x^{(1,N)}}^2\tr{V_x^{(1,N)}}^2+ 2\sum_{x\in\{0,1\}^{N}}\tr{M_x^{(1,N)}}\tr{V_x^{(1,N)}}+2^N\\&= \sum_{k=0}^1\left(\sum_{x\in\{0,1\}^{N-1}}\tr{M_{k}^0M_x^{(2,N)}}^2\tr{V_k^0V_x^{(2,N)}}^2- 2\sum_{x\in\{0,1\}^{N-1}}\tr{M_k^0M_x^{(2,N)}}\tr{V_k^0V_x^{(2,N)}}+\frac{1}{2}2^N\right)\\&=:\sum_{k=0}^1 C_k^1(V_k^1)
    \end{aligned}
\end{equation}
where in the second step we isolated the sum over the first index. Defining $\langle A|B\rangle =\tr{AB^{\dagger}}$ for two matrices $A,B$ as the Hilbert Schmidt inner product, we have $\tr{V_k^1V_x^{(2,N)}}^2=\langle V_k^1|(V_x^{(2,N)})^T\rangle\langle (V_x^{(2,N)})^T| V_k^1\rangle$, $\tr{V_k^1V_x}=\langle V_k^1|(V_x^{(2,N)})^T\rangle $, then  $C_k^1(X)$ has the promised form with
\begin{equation}
    A_k^1= \sum_{x\in\{0,1\}^{N-1}}\tr{M_k^1M_x^{(2,N)}}^2 |(V_x^{(2,N)})^T\rangle \langle(V_x^{(2,N)})^T|,\quad\quad |B_k^j\rangle =-2\sum_{x\in\{0,1\}^{N-1}}\tr{M_k^1M_x^{(2,N)}}|(V_x^{(2,N)})^T\rangle
\end{equation}
Both these expressions are efficiently computable via MPS by evaluating each matrix element in a basis using MPS contractions. $A_k^1$ is positive semi definite as for any $X$, $\langle X|A_k^1|X\rangle$ is a sum of squares.
\end{proof}

The local optimization can the be performed by minimizing a quadratic form. For completeness, we need to show that this optimization is well defined.
\begin{cor}
$C^j(X_0,X_1)$ has a unique stationary point $C^j(X_0^*,X_1^*)$, which is a minimum. $X_k^*$ is a solution of the linear system
\begin{equation}
    2A_{k}^j|X\rangle+|B_k^j\rangle=0.
\end{equation}
\end{cor}
\begin{proof}
All stationary points of $C_k^j$ must be minima, since its Hessian is proportional to $A_k^j$, which is positive semi-definite (that is, $C_k^j$ is convex).
Considering each matrix element in $X$ as a variable, we have
\begin{equation}
    \nabla C_k^j (X)= 2A_{k}^j|X\rangle+|B_k^j\rangle
\end{equation}
 Then, the solutions of the given linear system give the stationary points of $C_k^j$. We have to make sure that solutions exist and the quadratic form actually has a stationary point, i.e. that $|B_k^j\rangle$ is in the image of $A_k^j$. This follows from the fact that the cost function is bounded from below. If $|B_k^j\rangle=0$, $X=0$ is a solution, otherwise, suppose $|B_k^j\rangle\neq \mathrm{Im}(A_k^j)$, since $A_k^j\geq 0$, it is hermitian, then this implies $|B_k^j\rangle\neq \mathrm{Ker}(A_k^j)^{\perp}$, i.e. there exists $|X\rangle$ in the kernel of $A_k^j$ such that $\langle X|B_k^j\rangle:=c_X >0$. But then $C_j^k(-sX)= 2^{n-1}-sc_X<0$ for some real $s$ large enough, which is a contradiction as $C_j^k(X)\geq 0$. Finally, we need to show the solution is unique. If $X$ and $Y$ are both solutions, then $|T\rangle=|Y\rangle-|X\rangle\in \mathrm{Ker}(A_k^j)$. Since we just showed that $|B_k^j\rangle\in \mathrm{Im}(A_{k}^j) $, we have $\langle T|B_{k}^j\rangle=0$, then $C_k^j(Y)=C_k^j(X+T)=C_k^j(X)$, which shows the minimum is unique.
\end{proof}
The minimization algorithm is guaranteed to converge. However, unfortunately, we have no guarantees on the speed of convergence, or even of the existence of a low bond dimension minimum close to $0$ cost.

Before moving on to discussing the performance of this algorithm, we need to address one last point: the efficient computation of the cost function suffers from floating point precision problems, more specifically, it is subject to so-called \emph{catastrophic cancellation} when it is very close to $0$ and will become inaccurate. Catastrophic cancellation occurs when the required accuracy of the result of a computation is higher than the accuracy of the data used to obtain the result. Floating point numbers on a computer are only stored up to machine precision, that is, up to 15 significant digits. Therefore, if $a$ and $b$ are two floating point numbers of order, for example, $10^5$, they are only specified up to an accuracy of $\sim 10^{-10}$, and so is the number $a+b$. But if the true value of $a+b$ is smaller than $10^{-10}$, this means the computer will not be able to correctly compute it.
In particular, recall that the cost function is computed as follows
\begin{equation}
    \label{eq:cost_func}C_M(V)=\tr{\prod_{j=1}^N\sum_{k=0}^1{M_k^j}^{\otimes 2}\otimes {V_k^j}^{\otimes 2} } -2 \tr{\prod_{j=1}^N\sum_{k=0}^1{M_k^j}\otimes {V_k^j}}+2^N
\end{equation}
Let us examine in what form catastrophic cancellation can affect this computation: let $\chi$ be the bond dimension of the MPS representation $V$ of the inverse vector $v$, then the first term in \cref{eq:cost_func} is a sum of $\chi^2 4^{d-1}$ terms, the accuracy of the result is then bounded by 
\begin{equation}
\delta \|\prod_{j=1}^N\sum_{k=0}^1{M_k^j}^{\otimes 2}\otimes {V_k^j}^{\otimes 2}\|_{\max} \chi^2 4^{d-1},
\end{equation}
where $\|\cdot\|_{\max}$ is the absolute largest value across all elements in the matrix and $\delta=10^{-15}$ is machine precision. In practice we observe this algorithm to produce large numbers in the inverse MPS matrices resulting in numbers the order of $\sim 10^4$ in the final matrix used for the computation of the cost function, which means that, e.g. for $d=3, \chi=3$, the accuracy of the first term is $\sim 10^{-9}$ as opposed to machine precision, $\sim 10^{-15}$. This accuracy can get significantly worse as $\chi$ increases.

A possible way of alleviating this problem of catastrophic cancellation is to modify the cost function by introducing a regularization. In particular, we choose a regularization that guides the optimization towards a solution closer to a translationally invariant MPS. As argued earlier, we expect a translationally invariant solution to exist. As we observe, this choice of regularization also avoids dramatic changes in the tensors for the first few optimization steps. In practice, this also lets us avoid the appearance of very large numbers in the tensors.

More specifically, we add the following term to the cost function
\begin{equation}
    R(V)= \sum_{k=1}^N\sum_{i=0}^1\left\|V_k^i-\frac 1 N \sum_{j=1}^N V_j^i\right\|_2^2 
\end{equation}
In other words, this is the norm of the difference of each tensor with the average of all others with the same physical leg index, if $V$ is translationally invariant, $R(V)=0$. We control the regularization with a parameter $\alpha\geq 0$, so that the new cost function is
\begin{equation}
    C_{\alpha}(V)=C(V)+\alpha\, R(V).
\end{equation}
This corresponds to modifying the local quadratic forms as follows
\begin{equation}
    A_{k}^j\to A_{k}^j+\alpha\left(1-\frac 1 N \right)^2I, \qquad |B_k^j\rangle\to |B_k^j\rangle+2\alpha\frac 1 N\left(1-\frac 1 N \right) \sum_{i\neq k}|V_i^j\rangle \;. 
\end{equation}
All the previous convergence results apply for fixed $\alpha$. 

In practice, we observed that choosing $\alpha$ to decrease as the ansatz approaches a good inverse improves convergence, as then the updates are smaller and tend not to stray the MPS away from translational invariance as much as in the beginning. Choosing a smaller $\alpha$ then allows to focus on achieving a high-quality approximate inverse rather than minimizing the distance from translational invariance.
\begin{figure}[h]
    \centering
    \includegraphics[width=.5\textwidth]{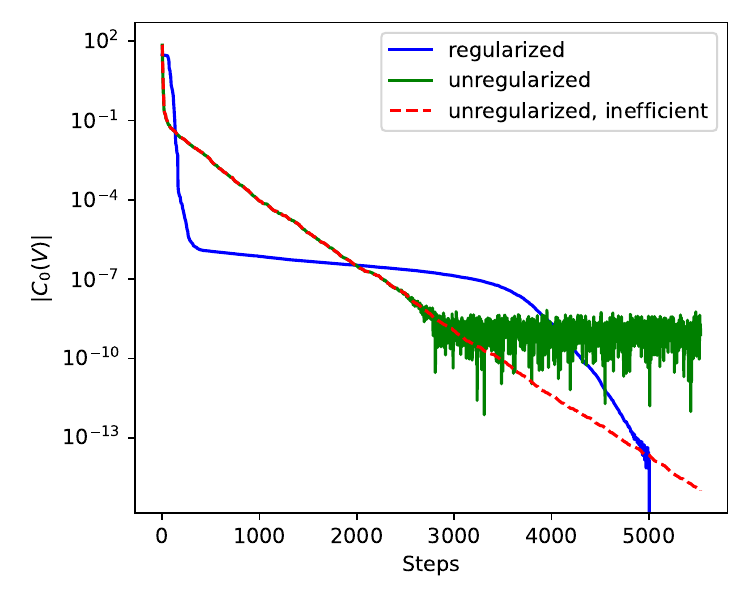}
    \caption{Cost function as a function of the number of steps for $n=10$, $d=3$, $\chi=3$.}
    \label{fig:inverse_error}
\end{figure}

\cref{fig:inverse_error} shows the cost function in \cref{eq:cost_function} as a function of the number of steps for $N=5$ ($n=10$ qubits), $d=3$, $\chi_V=3$ in three cases: first, we use no regularization and use efficient MPS contraction to compute the cost function. Here the floating point precision errors become significant and the MPS computation of the cost function becomes inaccurate below $\sim 10^{-9}$ as discussed earlier. To emphasize that this is an error in the MPS contraction, we show the cost function computed during the same unregularized optimization, computed inefficiently summing up each term in $\sum_x (m(x)v(x)-1)^2$. One can see that this way there is no precision problem and the cost function descends to machine precision. Lastly, we use the regularized cost function to perform the optimization, but we show $C_0(V)$ (since this is the measure which is actually important) computed using efficient MPS contractions, with the twist that at each step we pick $\alpha=C_0(V)$, since as discussed earlier this improves convergence. We can see that while the convergence is less regular, the floating point precision errors are not significant anymore.

Notice that other regularizations are possible, notably one can simply add terms $\alpha ||V_j^k||_2^2$ to discourage the algorithm to pick tensors which are too large in norm, which corresponds to the modification of the quadratic form $A_k^j\to A_{k}^j+I$. In practice we observed this to yield similar results.
In \cref{fig:algo_performance} we show the performance of the algorithm for various numbers of qubits and depths. 

\begin{figure}[h]
    \centering
    \includegraphics[width=\textwidth]{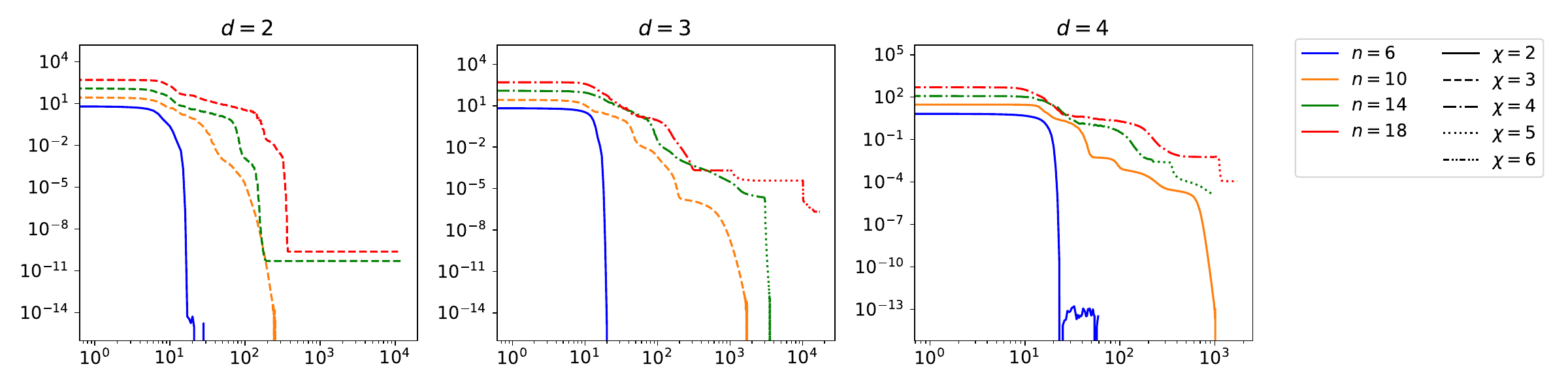}
    \caption{Cost function ($\alpha=0$) in the regularized optimization as a function of the number of steps. Here, a step is defined as the optimization of both indices corresponding to a single site. As indicated by the different line styles, sometimes we increase the bond dimension during the optimization to allow the algorithm to access better solutions. In these cases, we use the smaller bond dimension MPS as an ansatz for the higher bond dimension optimization.}
    \label{fig:algo_performance}
\end{figure}
\section{Robustness to approximate inverse}\label{app:proof_approx_m}
The map obtained via the procedure described in \cref{app:variational_inverse} is only an approximation of the real inverse map $\mathcal M^{-1}_d$. Here we show that if this approximation is close to the actual inverse measurement channel, the estimation of an expectation value obtained using it will be close to the real expectation value.

\begin{lem}
    Let $\mathcal{V}$ be a channel that is diagonal in the Pauli basis with coefficients $v_{\lambda}$ for each $P^{\lambda}\in \pauli{n}$ and is an approximate inverse of the measurement channel in the sense that there exists $\epsilon>0$ such that for all $\lambda\in \bitstring{2n}$,
    $\abs{1-\tprobl{\lambda}v_{\lambda}}\leq \epsilon$. Then
        \begin{equation}
        |\tilde o-\tr{O\rho}|\leq \epsilon||O||_{F}
    \end{equation}
    where $\tilde o$ is the expectation value obtained by using the approximate inverse channel in the shadow protocol.
\end{lem}
\begin{proof}
	By the Cauchy-Schwarz inequality,
	\begin{equation}
		|\tilde o-\tr{O\rho}|=|\tr{\rho(\mathcal V\circ \mathcal M(O)-O)}|\leq \|\mathcal V\circ \mathcal M(O)-O\|_F
	\end{equation}
For any $O=\sum_{\lambda}\alpha_{\lambda} P^{\lambda}$, we have
\begin{equation}
    \| \mathcal V \circ \mathcal M (O)-O\|_F^2=\left|\left|\sum_{\lambda} \alpha_\lambda P^{\lambda}\left(t_{\lambda,d}v_d-1\right) \right|\right|_F^2= \sum_\lambda \tr{I}\alpha_\lambda^2\left(t_{\lambda,d}v_d-1\right)^2\leq \epsilon^2 
    \|O\|_F^2,
\end{equation}
which proves the statement.
\end{proof}

\section{Proof of \cref{lem:ttprob_comp}}\label{app:proof_ttprob_MPS}

\begin{proof} Analogous to the tensor construction of $\tprobl{\lambda}$ given in \cref{fig:tprob_TN_MPS} and~\cref{app:proof_lem2}, here we will construct a tensor network with $4n$ uncontracted binary legs such that when these are assigned the input $\lambda, \lambda' \in \bitstring{2n}$, the tensor evaluates to $\ttprobl{\lambda, \lambda'}$.
We first define three tensors that will be the building blocks of $\ttprobl{\bullet}$.

First, we define $\Lambda$ a tensor with $4$ binary input and $4$ binary output legs (that specify a pair of input and a pair of output single qubit Paulis). For $g, g', h, h' \in \bitstring{2}$, we define $\Lambda$ by its entries 
\begin{align}
    \Lambda_{(h,h'),(g,g')} := \frac{1}{\abs{\clif{1}}}\sum_{U\in \clif{1}}\abs{\frac{1}{4}\tr{(P^{h}\otimes P^{h'})( UP^{g}U^{\dagger}\otimes UP^{g'}U^{\dagger})}}.
\end{align}
The entry $\Lambda_{(h,h'),(g,g')}$ can be interpreted as the probability that a pair of Paulis $(P^{g},P^{g'})$ are, under conjugation by $U$, mapped to the pair of Paulis $(P^{h},P^{h'})$ up to a factor of $\pm 1$ when $U$ is drawn uniformly at random from the set of single qubit Clifford gates $\clif{1}$. Note that if $d>0$, this layer does not change the final result (the measure $\mathcal U_d$ is invariant under single qubit Cliffords) and can be omitted in an actual construction.

We now define $\Gamma$, the two qubit analoge of $\Lambda$. Specifically, $\Gamma$ is a tensor with $8$ binary input and $8$ binary output legs (that specify a pair of input and a pair of output two qubit Paulis). For $g, g', h, h' \in \bitstring{4}$, we define $\Gamma$ by its entries 
\begin{align}
    \Gamma_{(h,h'),(g,g')} := \frac{1}{\abs{\clif{2}}}\sum_{U\in \clif{2}}\abs{\frac{1}{16}\tr{(P^{h}\otimes P^{h'})( UP^{g}U^{\dagger}\otimes UP^{g'}U^{\dagger})}}.
\end{align}
The entry $\Gamma_{(h,h'),(g,g')}$ can be interpreted as the probability that a pair of $2$-qubit Paulis $(P^{g},P^{g'})$ are, under conjugation by $U$, mapped to the pair of Paulis $(P^{h},P^{h'})$ up to a factor of $\pm 1$ when $U$ is drawn uniformly at random from the set of two qubit Clifford gates $\clif{2}$. 

By building a depth $d$ brickwork (c.f. \cref{fig:brickwork} for $d=3$) using $\Lambda$ tensors in the $0^{\mathrm{th}}$ layer and $\Gamma$ tensors in layers $1,2,\ldots,d$, we will construct a tensor with $4n$ input and $4n$ output legs that for $(\lambda,\lambda')\in \bitstring{4n}$ inputs and $(\nu,\nu')\in \bitstring{4n}$ outputs evaluates the probability that $U$ drawn from the ensemble  $\sensemble{d}$ will, up to a phase factor of $\pm 1$, map $P^{\lambda}$ to $P^{\nu}$ and simultaneously map $P^{\lambda'}$ to $P^{\nu'}$.
To compute $\ttprobl{\lambda,\lambda'}$ we simply need to sum over all $\nu,\nu'$ that are consistent with Pauli strings of $\mathbb{I}$ and $Z$. This can be implemented by a final layer of local $\Pi$ tensors attached to the brickwork.
We now define $\Pi$, a tensor with $4$ binary input (that specify a pair of single qubit Paulis) and no output legs. For $g, g' \in \bitstring{2}$, we define $\Pi$ by its entries 
\begin{align}
    \Pi_{(g,g')} := \begin{cases}
    1,& \text{if } g_x=g'_x=0\\
    0,              & \text{otherwise}
\end{cases},
\end{align}
where we use $g=:(g_x,g_z)$ and $g'=:(g'_x,g'_z)$ in line with the notation introduced above \cref{eq:pauli_index}. Intuitively, $\Pi_{(g,g')}$ evaluates to unity if both input Paulis have no Pauli $X$ component and zero otherwise. Hence, the tensor product of the local $\Pi$ tensors only evaluates to unity if the entire Pauli sting consists of a tensor product of $\mathbb{I}$ and $Z$ factors. We note that the $0^{\rm{th}}$ layer of $\Lambda$ tensors can be omitted when $d>0$ since it can be absorbed in to the $\Gamma$ tensor layer where it acts as a symmetry.

We note that by grouping tensor components in an analogous way to \cref{fig:tensor_definitions}, we can write $\ttprobl{\bullet}$ as an MPS with bond dimension $2^{4(d-1)}$. Some further reductions may be achivable but we leave this to future work. We note that in the $d=\order{\log n}$ setting, this tensor can be computed efficiently in $n$ as claimed.
\end{proof}

\section{Proofs of performance guarantees}\label{app:performance_guarantees}
First, we show that the state dependent and locally scrambled shadow norms are actually norms. 
\begin{lem}
For any state $\rho$
\begin{align*}
  \ssshadownorm{O}{d}{\rho}=\left(\expval{U\sim \sensemble{d}}{\sum_{b \in \bitstring{n}}
    \matrixel{b}{U \rho U^{\dagger}}{b} \matrixel{b}{U \invmeasmap{O} U^{\dagger}}{b}^2 }\right)^{1/2}
\end{align*}
is a norm.
\end{lem}
\begin{proof}
Clearly $\ssshadownorm{O}{d}{\rho}$ is non negative and vanishes for $O=0$. We note that
\begin{equation}
    \ssshadownorm{O}{d}{\rho}^2=\langle O,O\rangle_{\rho}
\end{equation}
where 
\begin{equation}
   \langle A,B\rangle_{\rho}= \expval{U\sim \sensemble{d}}{\sum_{b \in \bitstring{n}}
    \matrixel{b}{U \rho U^{\dagger}}{b} \matrixel{b}{U \invmeasmap{A} U^{\dagger}}{b}\matrixel{b}{U \invmeasmap{B^{\dagger}} U^{\dagger}}{b} }
\end{equation}
this expression is symmetric and sesquilinear, hence it is an inner product which induces the state dependent shadow norm, which in turns implies the triangle inequality.
\end{proof}
The following is then immediate by choosing $\rho=\mathbb{I}/2^n$:
\begin{cor}
$\ssshadownorm{O}{d}{\mathrm{LS}}$ is a norm.
\end{cor}
Below are some useful properties of $\tau_{(\lambda,\lambda')}$ which are self evident from the definition.
\begin{lem}\label{lem:ttprob_properties}
The following hold
\begin{enumerate}
    \item $0\leq \ttprobl{\lambda,\lambda'}\leq \min\{\tprobl{\lambda},\tprobl{\lambda'}, \tprobl{\lambda \oplus \lambda'}\}$,
    \item $\{P^{\lambda},P^{\lambda'}\}=0\implies
    \ttprobl{\lambda,\lambda'}=0$,
    \item $\ttprobl{\lambda,\lambda}=\tprobl{\lambda}$,
    \item $\ttprobl{\lambda,\lambda'}=\ttprobl{\lambda',\lambda}=\ttprobl{\lambda,\lambda \oplus \lambda'}$, and
    \item If for all $U\in \sensemble{d}$, $UP^{\lambda}U^{\dagger}$ and $UP^{\lambda'}U^{\dagger}$ have disjoint support, then $\ttprobl{\lambda,\lambda'}=\tprobl{\lambda}\tprobl{\lambda'}$.
\end{enumerate}
\end{lem}
We move on to the proof of \cref{lem:state_dep_snorm}.
\begin{proof}
Expand $O$ in the expression for the state-dependent shadow norm
\begin{equation}
\begin{aligned}
    2^{-n}\ssshadownorm{O}{d}{\sigma}&=\sum_{\lambda,\lambda'} \beta_{\lambda}\beta_{\lambda'}\frac{1}{t_{\lambda,d}t_{\lambda',d}} \expval{U\sim \mathcal U_d}{\bra{0} U\sigma U^{\dagger}\ket{0}\bra{0} UP^{\lambda} U^{\dagger}\ket{0}\bra{0} UP^{\lambda'} U^{\dagger}\ket{0}}\\&= 2^{-n} \sum_{\lambda,\lambda'} \beta_{\lambda}\beta_{\lambda'}\frac{1}{t_{\lambda,d}t_{\lambda',d}} \expval{U\sim \mathcal U_d}{\bra{0} UP^{\lambda} U^{\dagger}\ket{0}\bra{0} UP^{\lambda'} U^{\dagger}\ket{0}}\\&+2^{-n}\sum_{\lambda,\lambda'}\sum_{\lambda''\neq 0} \beta_{\lambda}\beta_{\lambda'}\tr{\sigma P^{\lambda''}}\frac{1}{t_{\lambda,d}t_{\lambda',d}} \expval{U\sim \mathcal U_d} {\bra{0} UP^{\lambda} U^{\dagger}\ket{0}\bra{0} UP^{\lambda'} U^{\dagger}\ket{0}\bra{0} UP^{\lambda''} U^{\dagger}\ket{0}}
    \end{aligned}
\end{equation}
where we have used $\sigma=2^{-n}\mathbb{I}+\sum_{\lambda\neq 0} \alpha_{\lambda}P^{\lambda}$. By \cref{lem:av_prod_paulis} and the same argument as in \cref{lem:measurement_map_action} we have
\begin{equation}
    \expval{U\sim \mathcal U_d}{\bra{0} UP^{\lambda} U^{\dagger}\ket{0}\bra{0} UP^{\lambda'} U^{\dagger}\ket{0}}= \delta_{\lambda,\lambda'} t_{\lambda, d}
\end{equation}
and again by \cref{lem:av_prod_paulis} the second term vanished unless $P^{\lambda''}=\pm P^{\lambda}P^{\lambda'}$. Then
we have
\begin{equation}
 \sum_{\lambda''\neq 0}\tr{\sigma P^{\lambda''}}\frac{1}{t_{\lambda,d}t_{\lambda',d}} \expval{U\sim \mathcal U_d} {\bra{0} UP^{\lambda} U^{\dagger}\ket{0}\bra{0} UP^{\lambda'} U^{\dagger}\ket{0}\bra{0} UP^{\lambda''} U^{\dagger}\ket{0}}= \tr{\sigma P^{\lambda}P^{\lambda'} } \tau_{(\lambda,\lambda'),d}
\end{equation}
with
\begin{equation}\label{eq:tau_exp_val}
    \tau_{(\lambda,\lambda'),d}=\mathbb E_{U\sim \mathcal U} (\bra{0} UP^{\lambda} U^{\dagger}\ket{0}\bra{0} UP^{\lambda'} U^{\dagger}\ket{0}\bra{0} UP^{\lambda}P^{\lambda'} U^{\dagger}\ket{0})=\prover{U\sim \sensemble{d}}{UP^{\lambda}U^{\dagger}, UP^{\lambda'}U^{\dagger}\in \pm\mathcal Z},
\end{equation}
where we have used \cref{lem:av_prod_paulis} again in the last equality. We then have
\begin{equation}
    \ssshadownorm{O}{d}{\sigma}^2=\sum_{\lambda}\frac{\beta_{\lambda}^2}{t_{\lambda,d}} +\tr{\sigma \sum_{\lambda\neq\lambda'}\beta_{\lambda}\beta_{\lambda'} \frac{\tau_{(\lambda,\lambda'),d}}{t_{\lambda,d}t_{\lambda',d}}P^{\lambda}P^{\lambda'}}=\sum_{\lambda}\frac{\beta_{\lambda}^2}{t_{\lambda,d}} +\tr{\sigma \tilde O}.
\end{equation}
It remains to identify the first term as the locally scrambled shadow norm, to see this, notice that 
\begin{equation}
    \expval{\sigma\sim \mathcal E}{\tr{\sigma \tilde O}}= \frac{1}{2^n}\tr{\tilde O}=0.
\end{equation}
\end{proof}
Next, we prove \cref{lem:shadow_norm_bounds}.

\begin{proof}
Recall that (cf. \cref{lem:state_dep_snorm})
\begin{equation}
    ||O||_{s(d),\sigma}^2=||O||_{s(d),\mathrm{LS}}^2+\tr{\sigma \tilde O}
\end{equation}
where
     \begin{equation}
        \tilde O=\sum_{\substack{\lambda', \lambda \in \bitstring{2n}\\\lambda'\neq\lambda}}P^{\lambda}P^{\lambda'} \beta_{\lambda} \beta_{\lambda'} \frac{\ttprobl{\lambda,\lambda'}}{\tprobl{\lambda} \tprobl{\lambda'}}.
    \end{equation}
    Notice that by Hölder's inequality $\tr{\sigma \tilde O}\leq ||\sigma||_1||\tilde O||_{\infty}=||\tilde O||_{\infty}\leq ||\tilde O||_F$, which proves
    \begin{equation}
        ||O||_{s(d)}^2\leq ||O||_{s(d),\mathrm{LS}}^2+||\tilde O||_F
    \end{equation}
    Next, we need to show that this expression can be computed efficiently via tensor contractions. Consider the operator
  \begin{equation}
        \tilde O'=\sum_{\lambda', \lambda \in \bitstring{2n}}P^{\lambda}P^{\lambda'} \beta_{\lambda} \beta_{\lambda'} \frac{\ttprobl{\lambda,\lambda'}}{\tprobl{\lambda} \tprobl{\lambda'}}.
    \end{equation}
The only difference between $\tilde O$ and $\tilde O'$ is that the latter includes terms with $\lambda=\lambda'$. We have
\begin{equation}
    ||\tilde O'||_F^2=||\tilde O||_F^2+2^n\sum_{\lambda\in \{0,1\}^n}\frac{\beta_{\lambda}^2}{t_{\lambda,d}}=||\tilde O||_F^2+2^n||O||_{s(d),\mathrm{LS}}^2
\end{equation}
and consequently
\begin{equation}
    ||\tilde O||_F^2= ||\tilde O'||_F^2- 2^n||O||_{s(d),\mathrm{LS}}^2
\end{equation}
    $||O||_{s(d),\mathrm{LS}}^2$ can be computed efficiently given an efficient MPS representation for $1/t_{\lambda,d}$ if $O$ is a shallow observable, as then we have access to an efficient MPS representation for $\beta_{\lambda}$, what's left is to compute $||\tilde O'||_F^2$. We do this by constructing an MPS representation of $\tilde O'$, after which the norm is easily computable. Recall that we can construct an MPS representation with polynomial bond dimension for $\tau_{(\lambda,\lambda'),d}$. Given access to MPS representations of $\beta_{\lambda}$, $\frac{1}{t_{\lambda,d}}$ of bond dimensions $\chi_{O},\chi_{\mathrm{inv}}$ respectively, can construct an MPS with bond dimension $\chi_{O}^2\chi_{\mathrm{inv}}^22^{4(d-1)}$ for
    \begin{equation}
        \beta_{\lambda}\beta_{\lambda'}\frac{\ttprobl{\lambda,\lambda'}}{\tprobl{\lambda} \tprobl{\lambda'}}.
    \end{equation}
    We then contract this with the local tensors
    \begin{equation}
        \Delta_{x,y,z}=\delta_{x, y\oplus z}
    \end{equation}
    which does not increase the bond dimension. The resulting MPS equals $\tilde O'$. See \cref{fig:snorm_bound_tensor} for a pictorial representation.
    \begin{figure}
        \centering
        \includegraphics[width=.5\textwidth]{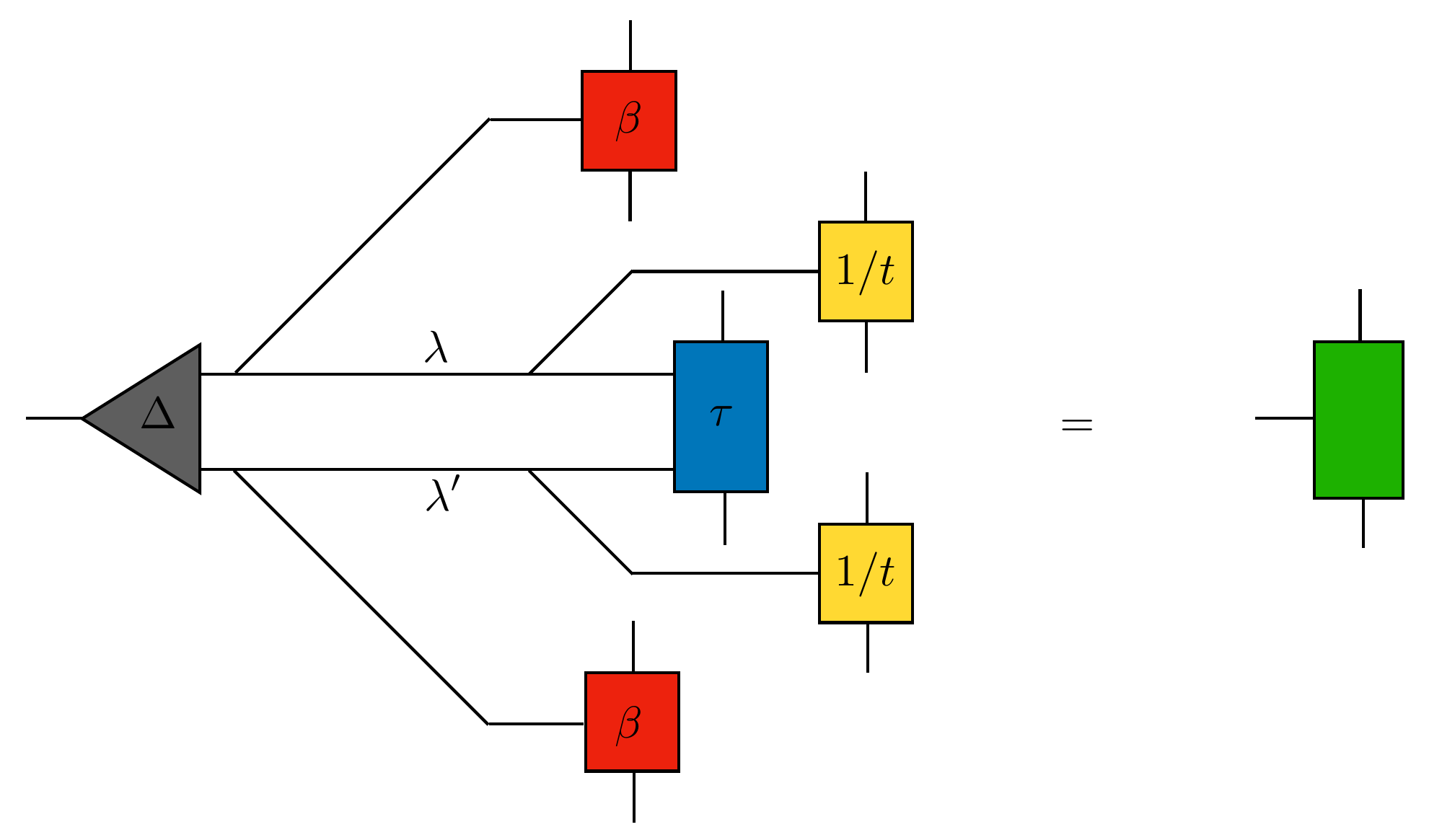}
        \caption{Construction of the local tensors for the MPS of $\tilde O'$.}
        \label{fig:snorm_bound_tensor}
    \end{figure}
\end{proof}
It is worth noticing that the Frobenius norm bound in \cref{eq:sn_frob_bound}  can also be obtained as a relaxation of the optimization over states obtained by dropping the positivity condition: let
\begin{align}
    \mathbb{A}^{+}&:=\set{\sigma \in {(\bbc^2)}^{\otimes n} ~|~\sigma=\sigma^{\dagger}, \tr{\sigma}=1, \tr{\sigma^2}\leq 1, \sigma\succeq 0},\\
    \mathbb{A}&:=\set{\sigma \in {(\bbc^2)}^{\otimes n} ~|~\sigma=\sigma^{\dagger}, \tr{\sigma}=1, \tr{\sigma^2}\leq 1}.
\end{align}
then
\begin{equation}
||O||_{s(d)}^2=||O||_{\mathrm{LS}}^2+\max_{\rho\in \mathbb A^+} \tr{\tilde O\rho}
\end{equation}
let us relax the optimization by removing the positivity condition
\begin{equation}
    \max_{\rho\in \mathbb A^+}\tr{\tilde O\rho}\leq  \max_{\rho\in \mathbb A}\tr{\tilde O\rho}
\end{equation}
Since $\mathbb A$ is convex, the optimizer must satisfy $\tr{\sigma_*}=\tr{\sigma_*^2}=1$. We can then take $\sigma=2^{-n}I+\sqrt{1-2^{-n}}A$ imposing $\tr{A}=0$, so that $\tr{\sigma}=1$, and $\tr{A^2}=1$, so that $\tr{\sigma^2}=1$. Then 
\begin{equation}
    \max_{\rho\in \mathbb A}\tr{\tilde O\rho}=\sqrt{1-2^{-n}}\max_{A}\tr{\tilde O A}
\end{equation}
the maximizer is then given by $A=\frac{\tilde O}{||\tilde O||_F}$ yielding
\begin{equation}
    \max_{\rho\in \mathbb A^+}\tr{\rho \tilde O} \leq \sqrt{1-2^{-n}} \frac{1}{||\tilde O||_F}\tr{\tilde O^2}= \sqrt{1-2^{-n}}||\tilde O||_F\leq ||\tilde O||_F.
\end{equation}
\section{The shadow norm of a stabilizer projector}\label{app:stab_state_snorm}
Here, we prove \cref{thm:stab_state_snorm}.
\begin{proof}
We begin by the following observation
\begin{lem}\label{lem:snorm_conjugation}
Let $\mathcal U=(\mathbb U,\mu)$ be an ensemble of unitaries where $\mathbb U$ is a subset of the unitary group on $n$ qubits and $\mu$ is a measure thereon. Let
\begin{equation}
    \mathcal M: \rho\mapsto \frac{1}{2^n}\sum_{b\in\{0,1\}^n}\mathbb E_{U\sim \mathcal U}(\langle b|U\rho U^{\dagger} |b\rangle U^{\dagger}|b\rangle\langle b|U )
\end{equation}
be the measurement channel obtained by this measure and let $||\cdot ||_{s,\rho}$ be the corresponding shadow norm.
Suppose $\mu$ is invariant under the application of unitaries from some subset $\mathbb V\subseteq \mathbb U$, for $V\in \mathbb V$
\begin{equation}
    ||VOV^{\dagger}||_{s,\rho}= ||O||_{s,V^{\dagger}\rho V}
\end{equation}
\end{lem}
\begin{proof}
Start by noticing that, by invariance under $\mathbb V$, 
\begin{equation}\label{eq:M_comm_V}
\begin{aligned}
    \mathcal M(VXV^{\dagger})&=\frac{1}{2^n}\sum_{b\in\{0,1\}^n}\mathbb E_{U\in \mathcal U} \langle b|UVXV^{\dagger}U^{\dagger}|b\rangle U^{\dagger} |b\rangle\langle b|U\\&=\frac{1}{2^n}\sum_{b\in\{0,1\}^n}\mathbb E_{U\in \mathcal U} \langle b|UXU^{\dagger}|b\rangle VU^{\dagger} |b\rangle\langle b|UV^{\dagger}=V\mathcal M(X) V^{\dagger}
    \end{aligned}
    \end{equation} 
    Then the same holds for the inverse: we have, using \cref{eq:M_comm_V},
    \begin{equation}
        \mathcal M(V\mathcal M^{-1}(X) V^{\dagger})= V\mathcal M(\mathcal M^{-1}(X))V^{\dagger}=VXV^{\dagger}
    \end{equation}
    applying $\mathcal M^{-1}$ to both sides we get
    \begin{equation}\label{eq:inv_comm_V}
        \mathcal M^{-1}(VXV^{\dagger})=V\mathcal M^{-1}(X)V^{\dagger}.
    \end{equation}
    By \cref{eq:inv_comm_V} we have,
    \begin{equation}
    \begin{aligned}
        ||VOV^{\dagger}||_{s,\rho}^2&=
    \expval{U\sim \mathcal U}{\sum_{b \in \bitstring{n}}
    \matrixel{b}{U \rho U^{\dagger}}{b} \matrixel{b}{U \invmeasmap{VOV^{\dagger}} U^{\dagger}}{b}^2 }\\&=\expval{U\sim \mathcal U}{\sum_{b \in \bitstring{n}}
    \matrixel{b}{U \rho U^{\dagger}}{b} \matrixel{b}{U V\invmeasmap{O}V^{\dagger} U^{\dagger}}{b}^2 }\\&=\expval{U\sim \mathcal U}{\sum_{b \in \bitstring{n}}
    \matrixel{b}{UV^{\dagger} \rho VU^{\dagger}}{b} \matrixel{b}{U \invmeasmap{O}U^{\dagger}}{b}^2 }= ||O||_{s,V^{\dagger}\rho V}^2.
    \end{aligned}
    \end{equation}
\end{proof}
In particular for the case at hand, the ensemble $\mathcal U_d$ is Pauli invariant. $\Pi=2^{-k}\sum_{s\in \mathcal S} s$ commutes with all stabilizers, hence we have
\begin{equation}
||\Pi ||_{s(d),\rho}^2=\frac{1}{2^k}\sum_{s\in \mathcal S} ||s\Pi s||^2_{s(d),\rho}=\frac{1}{2^k}\sum_{s\in \mathcal S} ||\Pi ||^2_{s(d),s\rho s}= ||\Pi ||_{s(d),\mathcal P(\rho)}^2
\end{equation}
where
\begin{equation}
    \mathcal P: \rho\mapsto \frac{1}{2^k}\sum_{s\in \mathcal S} s\rho s
\end{equation}
projects onto the commutant of $\mathcal S$, $\mathrm{Comm(\mathcal S})=\{X|[X,s]=0,\,\forall s\in \mathcal S\}$. In the second equality, we used \cref{lem:snorm_conjugation}, in the last one, we used the fact that $||\cdot||_{s(d),\rho}^2$ is linear in $\rho$. Then we have that
\begin{equation}
    ||\Pi||_{s(d)}^2=\max_{\rho}||\Pi||_{s(d),\rho}^2= \max_{\rho\in \mathrm{Comm(\mathcal S})}||\Pi||_{s(d),\rho}^2
\end{equation}
Since $\mathrm{Comm(\mathcal S})$ is convex and $||\Pi||_{s(d),\rho}^2$ is linear in $\rho$, the optimizer must be a pure state $|\psi\rangle\langle \psi| \in \mathrm{Comm(\mathcal S})$. But if $|\psi\rangle\langle \psi|$ commutes with all stabilizers this implies that for all $s\in \mathcal S$
\begin{equation}
    s|\psi\rangle=c(s)|\psi\rangle
\end{equation}
with $c(s)=\langle \psi|s|\psi\rangle=\pm 1$, hence the optimizer is a simultaneous eigenstate of all stabilizers. To finish, we need to prove that the largest value among these eigenvectors is achieved when $c(s)=1$ for all $s$. Let $\Lambda_{\mathcal{S}}=\{\lambda\in \{0,1\}^{2n}| P^{\lambda}\in \mathcal S\}$ and let $\tilde c({\lambda})=\langle \psi|P^{\lambda}|\psi\rangle$, then we have (cf. \cref{lem:state_dep_snorm})
\begin{equation}
    ||\Pi||_{s(d),|\psi\rangle\langle \psi|}^2=4^{-k}\sum_{\lambda,\lambda' \in \Lambda_{\mathcal S}}\langle \psi|P^{\lambda}P^{\lambda'}|\psi\rangle \frac{\ttprobl{\lambda,\lambda'}}{\tprobl{\lambda} \tprobl{\lambda'}} = 4^{-k}\sum_{\lambda,\lambda' \in \Lambda_{\mathcal S}}\tilde c(\lambda)\tilde c(\lambda')\frac{\ttprobl{\lambda,\lambda'}}{\tprobl{\lambda} \tprobl{\lambda'}}.
\end{equation}
Since $\frac{\ttprobl{\lambda,\lambda'}}{\tprobl{\lambda} \tprobl{\lambda'}}>0$, the maximum is attained then $c(\lambda)=1$ for all $\lambda \in \Lambda_{\mathcal{S}}$. Then the optimizer $|\psi\rangle$ is any state stabilized by $\mathcal S$ and as promised 
\begin{equation}
    ||\Pi||_{s(d)}^2=4^{-k}\sum_{\lambda,\lambda' \in \Lambda_{\mathcal S}} \frac{\ttprobl{\lambda,\lambda'}}{\tprobl{\lambda} \tprobl{\lambda'}}.
\end{equation}
\end{proof}
A natural question is why should stabilizers play a special role, if after all we can draw the local unitaries in the circuit from the full unitary group rather than just the Clifford group. As a matter of fact, when $d\to \infty$ and the measure becomes the uniform distribution on the Clifford group, the above theorem holds more generally, that is, the maximizer shadow norm of any projector is given by any state in the image of the projector. What breaks as $d$ is lowered is that the measure is no longer invariant under arbitrary rotations, but the Pauli invariance of the ensemble allows us to recover this result for stabilizer projectors. In fact, this theorem holds for any Pauli invariant ensemble by replacing $t_{\bullet,d}$ and $\tau_{(\bullet),d}$ by the appropriate corresponding quantities.

\section{Proof of \cref{thm:av_shadow_norm_UB}}\label{app:av_shadow_proof}
We will prove the following statement.
\begin{lem}\label{lem:tp_bound}
For $d\geq \Omega(\log(n))$, 
\begin{equation}
    t_{\lambda, d}\geq \frac{1}{2^{\min\{|\lambda|+2d, n\}}+1}\frac{1}{1+\frac{1}{n^{\mathcal O(1)}}}
\end{equation}
where $|\lambda|$ is the maximum distance between any two qubits where $P^{\lambda}$ is supported.
\end{lem}
Combining this with the expression for the shadow norm, we obtain \cref{thm:av_shadow_norm_UB}. More specifically, we have that if $d\geq \Omega(\log(n))$
\begin{equation}
    \ssshadownorm{O}{d}{\mathrm{LS}}^2=\sum_{\lambda}\frac{\beta_{\lambda}^2}{t_{\lambda_d}}\leq \left(1+n^{-\Omega(1)}\right)(2^{n}+1)\sum_{\lambda}\beta_\lambda^2=\left(1+n^{-\Omega(1)}\right)\frac{2^n+1}{2^n}||O||_{F}^2
\end{equation}
and if $d= \Theta(\log(n))$, $|\lambda|\leq \mathcal O(\log(n))$, we have
\begin{equation}
    \frac{1}{t_{\lambda,d}}\leq \frac{n^{\mathcal O(1)}}{1+\frac{1}{n^{\mathcal O(1)}}}=n^{\mathcal O(1)},
\end{equation}
hence if $O$ only contains Paulis such that $|\lambda|\leq \mathcal O(\log(n))$, we get 
\begin{equation}
    \ssshadownorm{O}{d}{\mathrm{LS}}^2=\sum_{\lambda}\frac{\beta_{\lambda}^2}{t_{\lambda_d}}\leq n^{\mathcal O(1)}\sum_{\lambda}\beta_\lambda^2=n^{\mathcal O(1)}2^{-n}||O||_{F}^2.
\end{equation}
To prove \cref{lem:tp_bound}, we rely on a well known mapping from moments of random quantum circuits to a path counting problem \cite{Nahum_2018,Zhou_2019,NHJ2019} which we only briefly summarize. For an in depth derivation, we direct the reader to Ref.~\cite{NHJ2019}. We have
\begin{equation}
    \mathbb E_{U\sim\mathcal U}\langle 0|UPU^{\dagger}|0\rangle^2 = \mathbb E_{U\sim\mathcal U} \langle P |^{\otimes 2} U^{\otimes 2}\otimes \bar{U}^{\otimes 2}|0\rangle^{\otimes 4}
\end{equation}
where $|P\rangle$ is the vectorization of the Pauli $P$ and $U$ is a random depth $d$ Clifford circuit. The average over the circuit is simply the average over the individual two-local gates, which can be seen as drawn from the Haar measure of the unitary group in dimension 4, $\mu_H(4)$, since the Clifford group forms a unitary 3-design \cite{Zhu_2017}. Using a tool known as \emph{Weingarten calculus} \cite{Brouwer_1996,Collins_2006}, an application of Schur-Weyl duality, we can write
\begin{equation}
    \mathbb E_{U\sim \mu_H(4)}U^{\otimes 2}\otimes \bar U^{\otimes 2}= \sum_{\sigma \in S_2} \mathrm{Wg}(\sigma^{-1}\tau, 4) |\sigma\rangle\langle\tau|  
\end{equation}
where $S_2$ is the permutation group of two elements, $\mathrm{Wg}$ is the Weingarten function, and for $\sigma \in S_2$ $|\sigma\rangle$ is the vectorization of the operator on $\mathbb C^4$ which permutes the 4 tensor product elements as $\sigma$.
Graphically, we can represent this equation as follows
\begin{equation}
    \mathbb E \begin{gathered}
    \includegraphics[height=1.5cm]{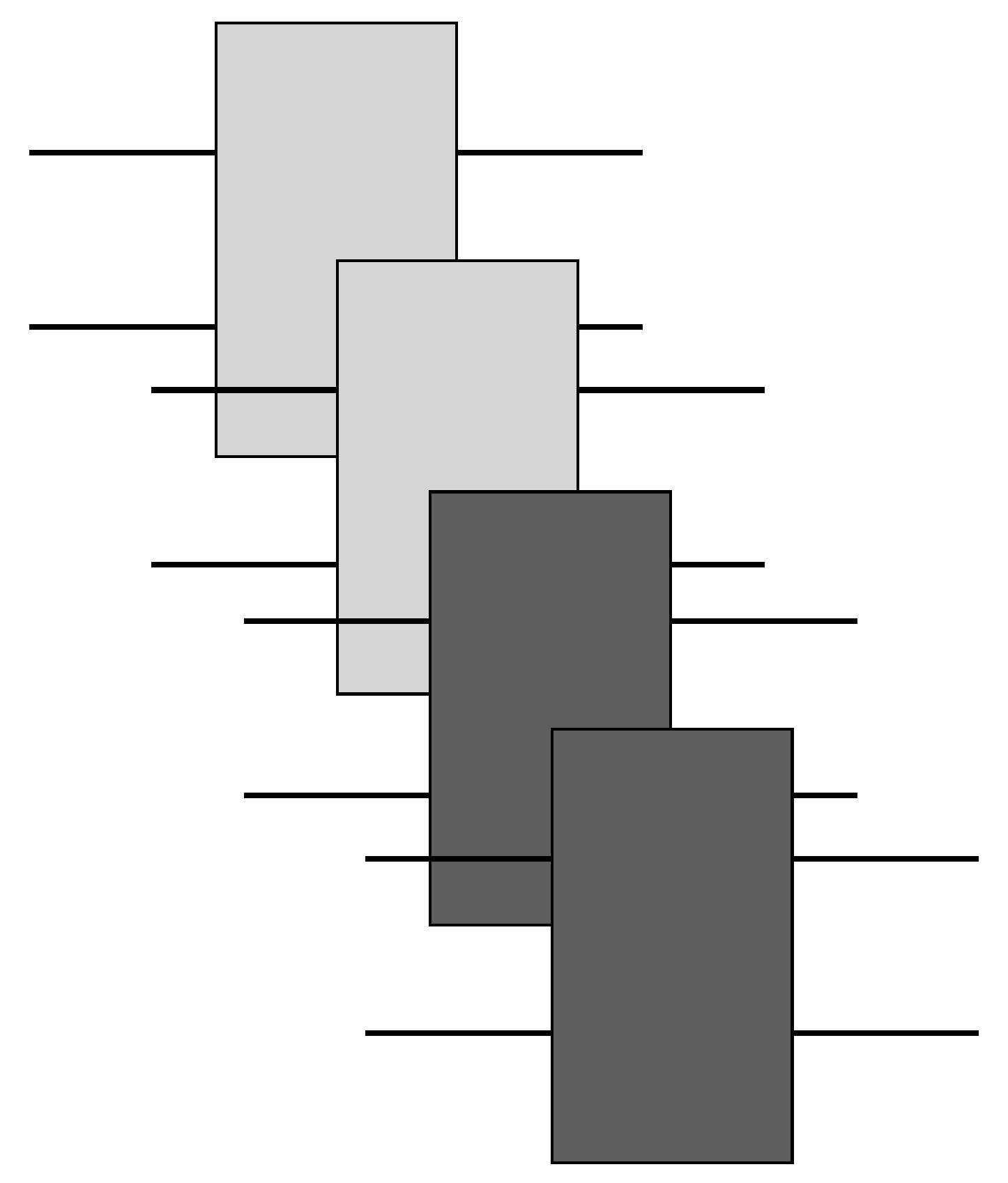}\end{gathered}=\begin{gathered}\includegraphics[height=1.5cm]{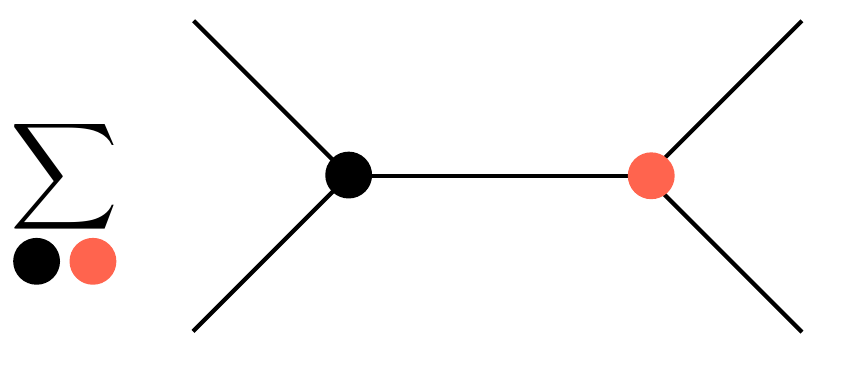}\end{gathered}.
\end{equation}
Here, the light boxes represent the two copies of $U$ and the dark boxes the two copies of $\bar U$, the black and red circles represent choices of permutations. The outgoing and ingoing legs are permutation vectorizations, the leg connecting the two permutation represents the Weingarten function $\mathrm{Wg}(\sigma^{-1}\tau,4)$. Applying this to the whole circuit, we get the sum over all black and red circles in the hexagonal lattice in Fig.~\ref{fig:hexagonallattice}.
\begin{figure}[h]
    \centering
    \includegraphics[width=.5\textwidth]{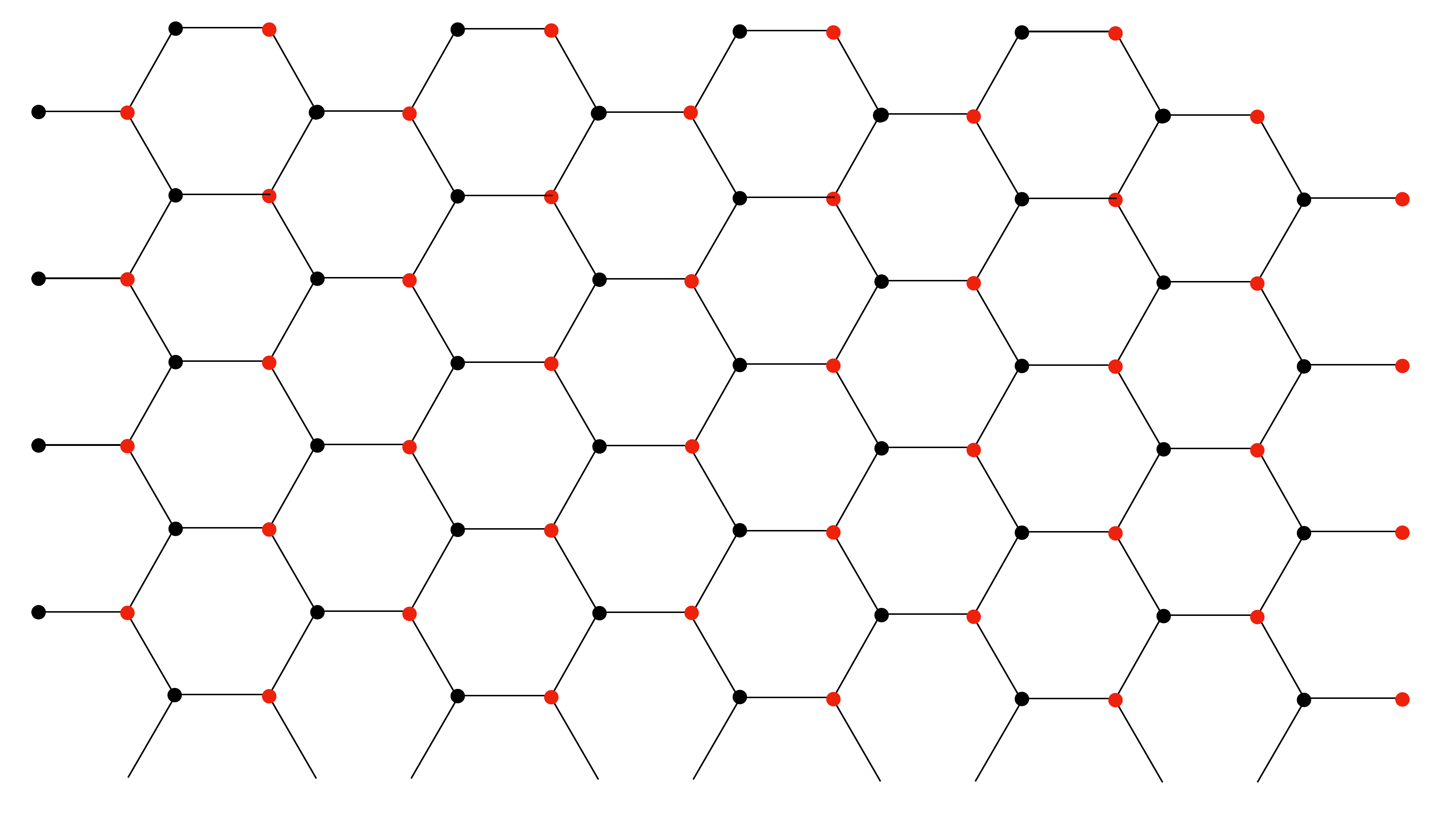}
    \caption{Summing over all black and red circles in this hexagonal lattice corresponds to averaging over all two-local gates in the brickwork random quantum circuit.}
    \label{fig:hexagonallattice}
\end{figure}
This is reminiscent of the partition function of a statistical mechanics model in which the spins are given by choices of permutations, and the interactions are given by either Weingarten functions or overlaps of permutation vectorizations,
performing this sum becomes substantially simpler if we first sum over all red circles, we then get an equivalent model on a triangular lattice
\begin{equation}
\begin{gathered}
    \includegraphics[height=1.5cm]{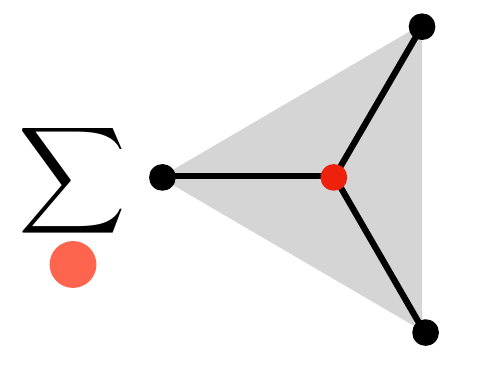}\end{gathered}=\begin{gathered}\includegraphics[height=1.8cm]{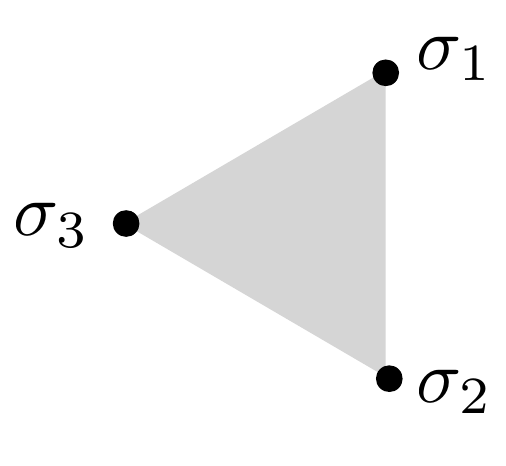}\end{gathered}= \begin{cases}
    1 \textrm{ if } \sigma_1=\sigma_2=\sigma_3 \\
    \frac 2 5 \textrm{ if } \sigma_1=\sigma_3,\, \sigma_2=\sigma_3\\
    0\textrm{ otherwise }
    \end{cases}
\end{equation}
sums over vertices in this new triangular model can be seen as sums over domain walls between regions of identity and flip permutations, each domain wall of length $l$ has a weight of $(2/5)^l$. Fig.~\ref{fig:path_example} shows the triangular lattice with an example of a two domain walls configuration
\begin{figure}[h]
    \centering
    \includegraphics[width=.4\textwidth]{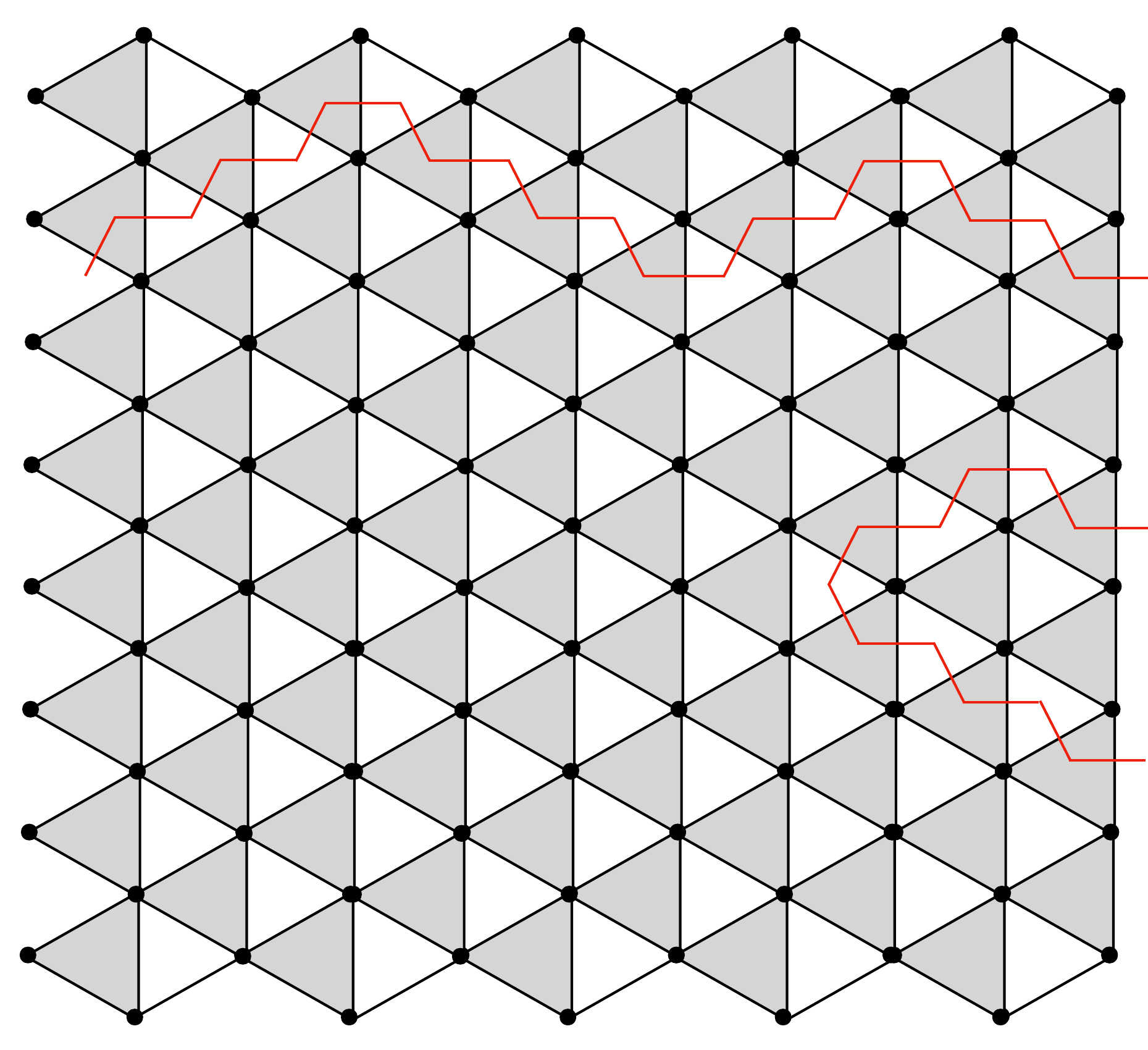}
    \caption{Triangular lattice with two domain walls.}
    \label{fig:path_example}
\end{figure}
We still have to discuss the boundary condition. On the right boundary every permutation is contracted with the vector $|0\rangle ^4$, and $\langle \sigma|0,0,0,0\rangle=1$. We still have to sum over the red circles at the right boundary, the value of this sum is independent of the neighboring black circle and from the whole right boundary, we get 
\begin{equation}
    (\mathrm{Wg}(\mathbbm 1,4)+\mathrm{Wg}(F,4))^{n/2}= \frac{1}{\sqrt{5}^n}
\end{equation}
for every boundary spin, where $\mathbbm 1, F\in S_2$ are the identity and flip permutation respectively.
The left boundary is more interesting. For a non trivial Pauli, we have
\begin{equation}
    \langle \mathbbm 1|(|P\rangle)^{\otimes 2}= \tr{P}^2=0 \quad \langle F|(|P\rangle)^{\otimes 2} =\tr{P^2}=2.
\end{equation}
For the identity, instead 
we have
\begin{equation}
    \langle \mathbbm{1}|(|\mathbbm{I}\rangle)^{\otimes 2}= 4\quad \langle F|(|\mathbbm{I}\rangle)^{\otimes 2} = 2.
\end{equation}
Hence, wherever there is a non trivial Pauli at the left hand boundary, only the flip permutation is allowed. We are now ready to lower bound $t_{\lambda,d}$. First of all, notice that if $P^{\lambda}=P\otimes \mathbb{I}$ is not fully supported
\begin{equation}
    t_{\lambda,d}=\expval{U\sim\mathcal U_d} {\langle 0|UP\otimes \mathbb{I} U^{\dagger} |0\rangle^2}\geq \expval{U\sim\mathcal U_d}{\langle 0|UP\otimes Q U^{\dagger} |0\rangle^2}
\end{equation}
for any Pauli $Q$, since $Q<\mathbb{I}$ as operators. Hence for a lower bound we just need to lower bound $t_{\lambda,d}$ for fully supported $P^{\lambda}$. We then have
\begin{equation}
    t_{\lambda,d}= \left(\frac{2}{\sqrt{5}}\right)^n\sum_{p\in \textrm{paths}(d)}\textrm{weight}(p)
\end{equation}
  since the left boundary is constrained to be comprised entirely of flip permutations, there cannot be any domain walls 
  reaching it, hence all the paths in the sum must begin and end at the right boundary. $\textrm{paths}(d)$ are all paths that do not reach beyond a distance $d$ from the right boundary. We have
\begin{equation}
\begin{aligned}
    t_{\lambda,\infty}&=\left(\frac{2}{\sqrt{5}}\right)^n\sum_{p\in \textrm{paths}(\infty)}\textrm{weight}(p)\leq \left(\frac{2}{\sqrt{5}}\right)^n\sum_{p\in \textrm{paths}(d)}\textrm{weight}(p)\left(1+\sum_{p\in \mathrm{paths}(> d)}\textrm{weight}(p)\right)\\&=:t_{\lambda,d}\left(1+\sum_{p\in \mathrm{paths}(> d)}\textrm{weight}(p)\right)
    \end{aligned}
\end{equation}
where $\mathrm{paths}(> d)$ are all paths that reach beyond a distance $d$ from the left boundary, and we have used that a configuration containing both paths that stay within a distance $d$ from the boundary and paths that do not has a weight equal to the product of the weights of these two paths. See \cref{fig:path_example_limit}.
\begin{figure}[h]
    \centering
    \includegraphics[width=.4\textwidth]{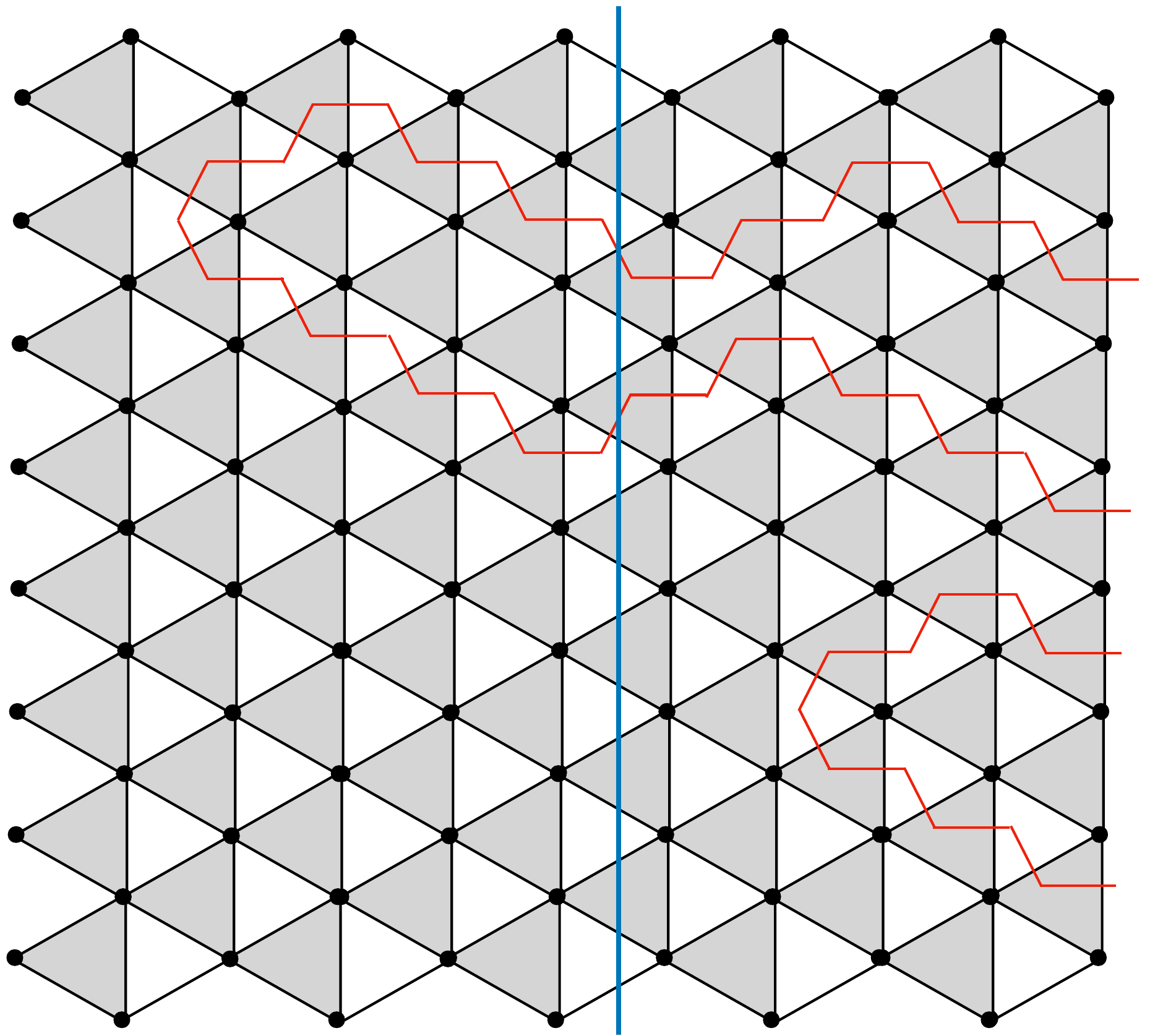}
    \caption{Two paths beginning and ending at the right boundary, the blue line represents the depth $d$, the longer path is not counted in $t_{\lambda,d}$, but it is in $t_{\lambda,\infty}$, where the weight of this configuration is the product of the weights of the two paths.}
    \label{fig:path_example_limit}
\end{figure}
We then need to upper bound this sum. We have
\begin{equation}
\sum_{p\in \mathrm{paths}(> d)}\textrm{weight}(p)=\sum_{t>d} R(t).
\end{equation}
Where $R(t)$ is the sum of all configurations that only contain paths that reach beyond a distance $t$ from the right boundary and contain at least one path reaching exactly a distance $t$. Consider all paths reaching some distance $t$ from the boundary: each of these paths can go either up or down at every of the $t$ steps from the boundary to the turn around point, and similarly for the way back to the boundary. Then for every choice of turn around point at a distance $t$ from the boundary there are fewer than $2^{2t}$ paths, since there are $n/2$ possible turn around points, we conclude that there are at most $n/2\, 2^{2t}$ of these paths, and each of them has a weight $(2/5)^{2t}$. We then have
\begin{equation}
    R(t)\leq \frac{n}{2}\left(\frac{4}{5}\right)^{2t} \sum_{m=0}^{\infty} \left(\frac{n}{2}\right)^m\left(\sum_{w=t}^{\infty}\left(\frac{4}{5}\right)^{2w}\right)^{m}.
\end{equation}
In the above expression the case $m=0$ represents one path reaching a distance $t$ from the boundary, and the other terms contain that one path plus $m$ other paths that may reach a distance $t$ or higher from the boundary. There can clearly be at most $n/2$ such non overlapping paths in a given configuration, but summing to infinity is good enough for our purposes, as the leftover sum will be exponentially small in $n$,
\begin{equation}\label{eq:R(t)}
\begin{aligned}
    R(t)&\leq \frac{n}{2}\left(\frac{4}{5}\right)^{2t} \sum_{m=0}^{\infty} \left(\frac{n}{2}\right)^m\left(\sum_{w=t}^{\infty}\left(\frac{4}{5}\right)^{2w}\right)^{m}=\frac{n}{2}\left(\frac{4}{5}\right)^{2t}  \sum_{m=0}^{\infty} \left(\frac{n}{2}\right)^m\left(\frac{4}{5}\right)^{2t}\left(\frac{(4/5)^{2t}}{1-(4/5)^{2}}\right)^{m}\\&=\frac{n}{2}\left(\frac{4}{5}\right)^{2t} \sum_{m=0}^{\infty} \left(\frac{25}{18}\left(\frac{4}{5}\right)^{2t}n\right)^m.
    \end{aligned}
\end{equation}
Now
suppose that
\begin{equation}\label{eq:d_greater_log}
    d\geq \frac{\alpha \log(n)+\log(1/c)}{\log(16/25)}
\end{equation} for some $\alpha>1, c>0$, then $\left(\frac{16}{25}\right)^t\leq cn^{-\alpha}$ and for $n> \left(\frac{25}{18}c\right)^{\frac{1}{\alpha-1}}$, we have
\begin{equation}
    R(t)\leq \frac{n}{2}\left(\frac{4}{5}\right)^{2t} \sum_{m=0}^\infty \left(c\frac{25}{18}\, n^{1-\alpha}\right)^m = \frac{n}{2}\left(\frac{4}{5}\right)^{2t} \frac{1}{1-c\frac{25}{18}\, n^{1-\alpha}}
\end{equation}
we then have 
\begin{equation}
\begin{aligned}
    \sum_{p\in \mathrm{paths}(> d)}\mathrm{weight}(p)&\leq\sum_{t> d}R(t)\leq \frac{1}{1-c\frac{25}{18}\, n^{1-\alpha}} \frac{25}{18} n\left(\frac{4}{5}\right)^{2(d+1)}\leq \frac{16}{25} \frac{1}{\frac{18}{25 c}n^{\alpha-1}-1\, } .
    \end{aligned}
\end{equation}
Putting everything together, we have
\begin{equation}
    t_{\lambda,d}\geq t_{\lambda,\infty}\frac{1}{1+\frac{16}{25}\frac{1}{\frac{18}{25 c}n^{\alpha-1}-1\, }} .
\end{equation}
Suppose now that $P^{\lambda}$ is fully supported only on a subregion of the chain of length $p<n-2d$. Let us make the notation more specific and call $t_{\lambda,d}(n)$ the value of $t_{\lambda,d}$ on a system of size $n$. In this case we can see that $t_{\lambda,d}(n)=t_{\lambda,d}(n')$ with $n'=p+2d<n$, as the unitaries acting outside the light-cone of a radius $d$ around the support of the Pauli play no role. The same exact argument as earlier applies, the only difference is the limiting value
\begin{equation}
   t_{\lambda,d}(n)= t_{\lambda,d}(n')\geq  \frac{1}{2^{p+2d}+1}\frac{1}{1+\frac{16}{25}\frac{1}{\frac{18}{25 c}n^{\alpha-1}-1\, }}
   .
\end{equation}
We can even make the bound better by noticing that the factor $n$ in $n^{1-\alpha}$ came from counting the number of possible turn around points for the paths, which in the case of $t_{\lambda,d}(n')$ is $n'$ instead. So if we have, in addition to \cref{eq:d_greater_log}, $p+2d\leq r\log(n)$, and hence $d=\Theta(\log n)$, we can replace $n^{\alpha-1}$ with $\frac{n^{\alpha}}{r\log(n)}$ in the bound above.

\bibliography{bibRef}

\end{document}